\documentclass{llncs}
\usepackage{graphicx}
\usepackage{paralist}
\usepackage{mathptmx}
\usepackage{enumitem}
\usepackage{color}
\usepackage{cite}
\usepackage{booktabs}
\usepackage{url}
\urldef\springerdoi\url{https://doi.org/10.1007/978-3-319-64203-1_15}
\usepackage[cmex10]{amsmath}
\usepackage[utf8]{inputenc}

\renewcommand{\thefootnote}{\fnsymbol{footnote}}

\newif\ifarxiv
\arxivtrue
\ifarxiv
\newcommand{\figbasename}{fig/maxcost}
\else
\newcommand{\figbasename}{fig/maxcost-ep}
\fi

\usepackage{comment}
\ifarxiv
\newenvironment{hideproof}{\begin{proof}}{\end{proof}}
\else
\excludecomment{hideproof}
\fi

\newcommand\blfootnote[1]{%
  \begingroup
  \renewcommand\thefootnote{}\footnote{#1}%
  \addtocounter{footnote}{-1}%
  \endgroup
}

\usepackage{amsfonts}
\usepackage{xspace}

\usepackage{fixltx2e}
\usepackage[caption=false,font=scriptsize]{subfig}

\usepackage[noend,ruled,linesnumbered]{algorithm2e}

\SetCommentSty{algcommentfont}

\newcommand{\pb}{{\sc{MSE}}\xspace}
\newcommand{\greedy}{{\sc{FillGreedy}}\xspace}
\newcommand{\lptbytypes}{{\sc{GreedyFor2Types}}\xspace}

\newcommand\fixme[1]{}

\usepackage{hyphenat}

\begin{document}
\mainmatter

\title{%
  Optimizing egalitarian performance\\in the side-effects model of colocation\\for data center resource management%
}

\titlerunning{Egalitarian performance in the side-effects colocation}

\author{%
Fanny Pascual\inst{1}
\and Krzysztof Rzadca\inst{2}}  
\institute{%
  Sorbonne Universit\'es, UPMC, % (Universit\'e Paris 6),
  LIP6, CNRS, UMR 7606, Paris, France\\
Email: \url{fanny.pascual@lip6.fr}
\and
Institute of Informatics, University of Warsaw, Poland\\
Email: \url{krz@mimuw.edu.pl}}

\maketitle

\begin{abstract}
In data centers, up to dozens of tasks are colocated on a single physical machine. Machines are used more efficiently, but tasks' performance deteriorates, as colocated tasks compete for shared resources. As tasks are heterogeneous, the resulting performance dependencies are complex. In our previous work~\cite{pascual2015partition}
we proposed a new combinatorial optimization model that uses two parameters of a task --- its size and its type --- to characterize how a task influences the performance of other tasks allocated to the same machine.

In this paper, we study the egalitarian optimization goal: maximizing the worst-off performance. This problem generalizes the classic makespan minimization on multiple processors ($P||C_{\max}$).
We prove that polynomially-solvable variants of $P||C_{\max}$ are NP-hard and hard to approximate when the number of types is not constant.
For a constant number of types, we propose a PTAS, a fast approximation algorithm, and a series of heuristics.
We simulate the algorithms on instances derived from a trace of one of Google clusters.
Algorithms aware of jobs' types lead to better performance compared to algorithms solving $P||C_{\max}$.

The notion of type enables us to model degeneration of performance caused by colocation 
using standard combinatorial optimization methods. Types add a layer of additional complexity.
However, our results --- approximation algorithms and good average-case performance --- show that types can be handled efficiently.
\end{abstract} 

\section{Introduction}
The back-bone of cloud computing, the modern data center redefines how industry and academia use computers.
% \ifarxiv
% \blfootnote{Author's version of a paper published in Euro-Par 2017 proceedings. In F.F. Rivera et al. (Eds.): Euro-Par 2017, LNCS 10417, pp. 1–14, 2017, Springer.
%   The final publication is available at Springer via \url{https://link.springer.com/chapter/10.1007%2F978-3-319-64203-1-15}}
% \fi
\blfootnote{Author's version of a paper published in Euro-Par 2017 Proceedings; extends the paper with addtional results and proofs. In F.F. Rivera et al. (Eds.): Euro-Par 2017, LNCS 10417, Springer.
  The final publication is available at Springer via
  \springerdoi
}

Resource management in data centers significantly differs from scheduling jobs on a typical HPC supercomputer.
First, the workload is much more varied~\cite{reiss2012heterogeneity}:
data centers act as a physical infrastructure providing virtual machines,
or higher-level services, such as memory-cached databases
or network-in\-ten\-sive servers; in contrast, there are relatively few HPC-like com\-pu\-ta\-tio\-nally-intensive batch jobs (later, we will use a generic term \emph{task} for all these categories).
Consequently, a task usually does not saturate the resources of a single node~\cite{kambadur2012measuring}.
Tasks' loads vastly differ: in a published trace~\cite{reiss2012heterogeneity}, tasks' average CPU loads span more than 4 orders of magnitude.
In contrast to HPC scheduling in which jobs rarely share a node, heterogeneity in both the type and the amount of needed resources makes it reasonable to allocate multiple tasks to the same physical machine.

Tasks colocated on a machine compete for shared hardware. Despite significant advances in both OS-level fairness and VM hypervisors, virtualization is not transparent: multiple studies show \cite{kambadur2012measuring, koh2007analysis,  xu2013bobtail, kim15corunneraff, podzimek15cpupinning} that the performance of colocated tasks drops. Suspects include difficulties in sharing the CPU cache or the memory bandwidth.
The resource manager should thus colocate tasks that are compatible, i.e., that use different kinds of resources --- hence, it should optimize tasks' performance. This, however, requires a performance model.

Our side-effects model~\cite{pascual2015partition} bridges the gap between colocation in datacenters and the theoretical scheduling, bulk of which has been developed for non-shared machines.

Rather than trying to predict tasks' performance from OS-level metrics, we abstract by characterizing a task by two characteristics: \emph{type} (e.g.: a database, or a com\-pu\-ta\-tio\-nally-intensive job) and \emph{load} relative to other tasks of the same type (e.g.: number of requests per second).
The total load of a machine is a vector: its $i$-th dimension is the sum of loads of tasks of the $i$-th type located on this machine.
Each type additionally defines a performance function 
mapping this vector of loads to a type-relevant performance metric.
As datacenters execute multiple instances of tasks such function can be inferred by a monitoring module~\cite{xu2013bobtail,kim15corunneraff, podzimek15cpupinning} matching task's reported performance (such as the 95th percentile response time) with observed or reported loads.

In this paper, we consider optimization of the worst-off performance (analogous to makespan in classic multiprocessor scheduling problem, $P||C_{\max}$~\cite{graham1969bounds}). We use a linear performance function: on each machine, the influence a type $t'$ has on type $t$ performance is a product of the load of type $t'$ and a coefficient $\alpha_{t',t}$. The coefficient $\alpha_{t',t}$ describes how compatible $t'$ load is with $t$ performance (the coefficient is similar to interference/affinity metrics proposed in \cite{podzimek15cpupinning,kim15corunneraff}). Low values ($0 \leq \alpha_{t',t} < 1$) correspond with compatible types (e.g.: colocating a memory-intensive and a CPU-intensive task): it is preferable to colocate a task $t$ with tasks of the other type $t'$, rather than with other tasks of its own type $t$. High values ($\alpha_{t'', t}>1$) denote types competing for resources.
%, i.e., less incentive to colocate (at least from the tasks' owner's point of view).

The contribution of this paper is as follows.
\begin{inparaenum}[(1)]
\item We prove that the notion of type adds complexity, as makespan minimization with unit tasks $P|p_i=1|C_{\max}$ (a polynomially solvable variant of $P||C_{\max}$) becomes NP-hard and hard to approximate when the number of types $T$ is not constant (Section~\ref{sec:complexity}).
We then show how to cope with that added complexity. 
We propose 
\item a PTAS for a constant $T$ and a constant $\alpha$ (Section~\ref{sec:ptas}); and 
\item a fast greedy approximation algorithm (Section~\ref{sec:greedy-list}).
\item We also propose natural greedy heuristics (Section~\ref{sec:heuristics})
\ifarxiv
To characterize the optimal schedules in function of the coefficient $\alpha$, we study a series of special cases with two types ($T=2$, Section~\ref{sec:two-types}).
\item We identify two tipping points, i.e., values of $\alpha$ for which the shape of the optimal schedule changes. For $0 \leq \alpha \leq 1$, all machines should be shared between types (if possible). For $1 < \alpha < 2$, there are some instances that share all machines, but for divisible load (i.e., many small tasks), there is at most one shared machine. Finally, for $\alpha \geq 2$ at most one machine is shared. (Proposition~\ref{prop:clashing-1mach}). For each case
\item we show fast approximation algorithms.
  \else (in the accompanying technical report~\cite{pascual2017maxcost-techreport} we show they are approximations for $T=2$).
  \fi
\item We also test our algorithm by simulation on a trace derived from one of Google clusters (Section~\ref{sec:experiments}).
\end{inparaenum}

\section{Side-Effects of Colocating Tasks: A Model}\label{sec:model}
%\fixme{steady state, no notion of time; this is important but orthogonal; it's as normal scheduling that is also steady state; all tasks are ready to be executed at the beginning} DONE
%\fixme{K to F: do we keep ``size'' of tasks (in contrast to ``length'')? we should, as we use two interpretations of the model; if we keep ``size'', we should replace all short by small and all long by large. F to K : I think we should indeed keep size instead of length (because we do assignment of tasks to the machines - we do not consider that there is a duration for each task), so we should indeed replace all short by small and all long by large. } -> OK

% \fixme{after writing theory, return here and check whether we need all these notations} DONE

%\fixme{we use partition here with symbol $P$; but later on we use schedule with symbol $\sigma$} OK, done
We study a min-max (egalitarian) performance criteria for our side-effects performance model (introduced in~\cite{pascual2015partition}, where we studied a utilitarian objective, min-sum). We consider a system that allocates $n$ tasks $J=\{1, \dots, n\}$ to $m$ identical machines $\mathcal{M} = \{M_1,\dots, M_m\}$. Each task $i$ has a known size $p_i\in \mathbb{N}$ (i.e., clairvoyance, a  common assumption in scheduling; the sizes can be estimated by previous instances or users' estimates).
The size corresponds to the load the task imposes on a machine: the request rate for a web server; or the cpu load for a cpu-intensive computation.
We take other assumptions standard in scheduling theory: all tasks are known (off-line) and ready to be scheduled (released at time 0).
We take these assumptions to derive results on the basic model before tackling more complex ones. 
We denote by $p_{\max}=\max p_i$ the largest task and by $W$ the total load, $W=\sum p_i$.
We assume that the tasks are indexed by non-increasing sizes: $p_1\geq p_2\geq \dots \geq p_n$. 

%\fixme{we should also see whether we keep the term "schedule" - F: it can be useful, but in this case we should say somewhere that a partition can be called a schedule - since there are strong connections between our pb and scheduling pbs or algorithms (e.g. for the ptas, or for the greedy algorithm which looks like a list scheduling algorithm). We should also say that sometimes ``length" will refer to "size" (it is the case for the ptas). }OK done

A \emph{partition} (an \emph{allocation}) is an assignment of each of the $n$ tasks to one of the $m$ machines. A partition separates the tasks into at most $m$ subsets: each subset corresponds to the tasks allocated on the same machine. 
Given a partition $P$, we denote by $M_{P,i}\in \mathcal{M}$ the machine on which task $i$ is allocated. Due to the similarities with $P||C_{max}$, we sometimes use the term ``schedule" (and the symbol $\sigma$) for an allocation (and even the term of length for the size of a task). In this case, only the allocation is meaningful (not the order of the tasks on the machines).
% is not meaningful., the only thing that matters is the assignment of teh atsks to teh machines.
% BTW: I would replace all the "processor" by "machine" - done

The main contribution of this paper lies in analyzing side-effects of colocating tasks. 
The impact of task $i$ on the performance of another task $j$ is a function of task's size $p_i$ and \emph{task's type} $t_i$. 
Types generalize tasks' impact on the performance and may have different granularities: for instance, ``a webserver'' and ``a database''; or ``a read-intensive MySQL database''; or, as in~\cite{kim15corunneraff}, ``an instance of Blast'' .
We assume that the type $t_i$ is known (which again corresponds to the clairvoyance assumption in classic scheduling; typically a data~center runs many instances of the same task, so task's type can be derived from the past).
Let $\mathcal{T}=\{1,\dots,T\}$ be a set of $T$ different types of tasks. Each task $i$ has type $t_i \in \mathcal{T}$. 
For each type $t\in \mathcal{T}$, we denote by $J^{(t)}$
% \fixme{$J^t$ or $J^{(t)}$}
the tasks which are of type $t$;
and by $p_i^{(t)}$ the size of the $i$-th largest task of type $t$ (ties are broken arbitrarily).

% \fixme{2 interpretations of the cost: either completion time of the task with size adjusted by the previously-executed tasks --- thus extension of $P||C_{\max}$; or all the tasks share a resource and their performance degrades --- thus min-max variant of koutsoupias. \\
% F to K: for the first interpretation,  I would just say that for identical types (or for the single type case), our problem is equivalent to  $P||C_{\max}$. Indeed the cost of each task is the load of the machine on which it is, and thus the maximum cost is the maximum load, that is the makespan. }

We express performance of a task $i$ by a cost function $c_i$:
to simplify presentation of our results, we prefer to express our problems as minimization of costs, rather than maximization of performance (for a single type, our cost is synonymous with the makespan). Note that the cost is unrelated to monetary cost (the amount of money that a job pays to the machine) --- we do not consider monetary costs in this paper.
%Our cost model has two interpretations: we introduce the model using a data~center interpretation, and then show how the model relates to a classic multiprocessor scheduling problem $P||C_{\max}$.
%In the data~center interpretation,
Task's $i$ cost $c_i$ depends on to the \emph{total load} of tasks $j$ colocated on the same machine $M_{P,i}$, but different types have different impacts:%\vspace{-.8em} 
% \begin{equation}\label{eq:lin-cost}
$c_i = \sum_{j \mbox{ on machine } M_{P,i}} p_j . \alpha_{t_j,t_i},$%\vspace{-.8em} 
% \end{equation}
The cost function also takes into account the task $i$ itself, as well as other tasks of the same type.
A coefficient $\alpha_{t, t'}\in \mathbb{R}_{\geq 0}$
%\fixme{Krz to F: check! F to K:  I indeed think that the only thing we need is that $\alpha_{t, t'}\geq 0$. On wikipedia (https://en.wikipedia.org/wiki/Real\_number), the notation $ R+ \cup {0}$ is used instead of  $\mathbb{R^{0+}}$ : are you sure that this is the correct notation (I didn't know this notation before)?} OK thanks
defined for each pair of types $(t,t')\in \mathcal{T}^2$, 
measures the impact of the tasks of type $t$ on the cost of the tasks of type $t'$ (allocated on the same machine). 
If $\alpha_{t,t'}=0$ then a task of type $t$ has no impact on the cost of a task of type $t'$; the higher the $\alpha_{t,t'}$, the larger the impact. Coefficients are not necessarily symmetric, i.e., it is possible that $\alpha_{t,t'} \neq \alpha_{t',t}$.
The coefficients $\alpha_{t, t'}$ can be estimated by monitoring tasks' performance in function of their colocation and their sizes (a data~center runs many instances of similar services~\cite{kim15corunneraff,podzimek15cpupinning,xu2013bobtail}).
We consider the linear cost function as it generalizes, by adding coefficients $\alpha_{t,t'}$, the fundamental scheduling problem $P||C_{\max}$~\cite{graham1969bounds} (if $\forall (t,t')\in\mathcal{T}^2: \alpha_{t, t'}=1$, our problem reduces to $P||C_{\max}$).
Assuming linearity is a common approach when constructing models in operational research or statistics (e.g. linear regressions). Likewise, in selfish load balancing games~\cite{koutsoupias_worst-case_1999}, it is assumed that the cost of each task is the total load of the machine (but their model does not consider types).
We assume that the impact the type has on itself is \emph{normalized} with regards to tasks' sizes, i.e., $\alpha_{t,t}=1$  (although some of our results, notably the PTAS, do not need this assumption).

We denote by \pb ({\sc MinMaxCost with Side Effects}) the problem of finding a partition $P^*$ minimizing the maximum cost $C(P) = \max_i c_i$, with $c_i$ defined by the linear cost function. The partition $P^*$ minimizes the worst performance a task experiences in the system, thus corresponds to the egalitarian fairness.

%%%%%%%%%%%%%%%%%%%%%%%%%%%%%%%%%%%%%%%%%%%%%%%%%%%%%%%%%%%%%%%%%%%%%%%%%%%%%%%%%%%
%%%%%%%%%%%%%%%%%%%%%%%%%%%%%%%%%%%%%%%%%%%%%%%%%%%%%%%%%%%%%%%%%%%%%%%%%%%%%%%%%%%
% \section{Many types}
% In this section, we study a general version of \pb with $T$ types. First, we show that \pb is NP-complete (Proposition~\ref{prop:np-complete}) and hard to approximate (Proposition~\ref{prop:inapprox}) if the number of types and machines is not fixed. Then, we propose a PTAS for a constant number of types (Section~\ref{sec:ptas}) and a fast greedy approximation algorithm (Section~\ref{sec:greedy-list}).

\section{Complexity and hardness of \pb for $T$ not fixed}\label{sec:complexity}

\pb is NP-hard as it generalizes an NP-hard problem $P||C_{\max}$ when there is only one type.
Our main result is that a polynomially-solvable variant of mupltiprocessor scheduling ($P|p_i=1|C_{\max}$) becomes NP-hard when tasks are of different types. 
Thus types add another level of complexity onto an already NP-complete $P||C_{\max}$.

\begin{proposition}\label{prop:np-complete}
The decision version of \pb is NP-complete, even if all the tasks have unit size, and even if $m=2$. 
\end{proposition}

\begin{proof}
  (Sketch) Reduction from {\sc Partition}~\cite{garey_computers_1979}. Given a set $S = \{ a_i \}$ of $n$ positive integers summing to $2B$, we build an instance of \pb with $n$ tasks, each of size 1 and each of a different type. For a task $i$, we set its coefficients $\forall j: \alpha_{i,j} = a_i$. Partition of $S$ into two sets each with sum $B$ exists if and only if there exists an allocation $P$ with maximal cost $B$: cost of each task $j$ allocated to a machine $k$ is equal to $c_j = \sum_{i : M_{P,i}=k} \alpha_{i,j} = \sum_{i \in S_k} a_i$.
\end{proof}

% \subsection{Hardness of Approximation for $T$ not fixed}
%If the number of types is not fixed, \pb is hard to approximate within a constant factor even for unit-size jobs (again, in contrast to polynomially-solvable $P|p_i=1|C_{\max}$).

\begin{proposition}
\label{prop:inapprox}
\pb is strongly NP-hard, even if all tasks have unit size. Moreover, there is no polynomial time $r$-approximate algorithm for \pb, for any number $r>1$, unless $P=NP$.
\end{proposition}

\begin{proof}
Let $r>1$. We show that if there is a $r$-approximate algorithm for \pb, the algorithm solves NP-complete   {\sc Partition into cliques, PIC}~\cite{garey_computers_1979}.
In {\sc PIC}, given a graph $G=(V,E)$ and a positive integer $K\leq |V|$, can the vertices of $G$ be partitioned into $k\leq K$ disjoint sets $V_1, V_2, \dots, V_k$ such that, for $1\leq i \leq k$, the subgraph induced by $V_i$  is a complete graph? We assume that $V$ are labeled from $1$ to $|V|$.

% The decision version of our problem is the following one.  We are given a  bound (positive integer) $B$,  and an instance of \pb (i.e. a number $m$ of machines, a number $T$ of types, a matrix of size $T^2$ indicating for each couple of types $(i,j)$ the value $\alpha_{i,j}$, and a set of $n$ tasks such that each task $i\in\{1,\dots,n\}$ has a length (positive integer) $p_i$ and a type $t_i\in\{1,\dots, T\}$. The question is: is there a partition of the tasks into at most $m$ subsets such that the maximum cost of a task is at most $B$? 

 %Given an instance of  the   {\sc Partition into cliques} problem (a graph $G=(V,E)$ and a positive integer $K\leq |V|$), we create $K$ instances of (the decision version of) \pb.  Let $i\in\{1,\dots, K\}$. The $i^th$ instance of \pb is as follow: the number of machines is $m=i$; there are $n=|V|$ tasks, each of a different type. For each couple of types  $(i,j)$, we have, if $i\neq j$: $\alpha_{i,j}= 0$ if $\{i,j\}\in V$ and $\alpha_{i,j}= r$ if $\{i,j\}\notin V$; $\alpha_{i,i}=1$. All the tasks are of length $1$. We fix $B=1$.
 
 %The solution for the instance of the   {\sc Partition into cliques} problem is ``yes'' if and only if the solution is ``yes'' for (at least) one of the $|V|$ instances of \pb.
 
Given an instance of  {\sc PIC},
%(a graph $G=(V,E)$ and a positive integer $K\leq |V|$),
we create $K$ instances of  \pb.  Let $i\in\{1,\dots, K\}$. The $i$-th instance of \pb is as follows: the number of machines is $m=i$; there are $n=|V|$ tasks, each of a different type (types are labeled from 1 to $|V|$).
All the tasks are of size $1$.
For each type $i$, $\alpha_{i,i}=1$.
For each pair of types $(i,j), i\neq j$: $\alpha_{i,j}= 0$ if $\{i,j\}\in E$ and $\alpha_{i,j}= r$ if $\{i,j\}\notin E$.

We claim that a solution of a \pb instance costs either $1$ or at least $r+1$. 
%We also claim that the optimal cost of each of these \pb instances is either $1$, or at least $r+1$.
We also claim that the answer for the instance of {\sc PIC} is ``yes'' if and only if the optimal cost of one of these \pb instances is 1.
Therefore, an $r$-approximate algorithm for \pb will find a solution of cost 1 if it exists
(when there is a solution of cost 1, an $r$-approximate algorithm has to return a solution of cost at most $r$, which is thus necessarily the optimal solution since all the other solutions have a cost of at least $r+1$).
%This algorithm will thus solve NP-complete  {\sc Partition into cliques}.
Since $K\leq |V|$, if we assume that our $r$-approximate algorithm runs in polynomial time, then by using it $K$ times we can solve in polynomial time {\sc PIC}, which is an NP-complete problem. This leads to a contradiction, unless $P=NP$.
%Thus this algorithm cannot run in polynomial time, unless $P=NP$.  
 
We show that the cost of a solution of each of the \pb instances is either $1$, or at least  $r+1$.
%\fixme{K to F: I made the previous claim stronger, please check}
%We recall that all the tasks are of length 1, that $\alpha_{i,i}=1$, and that for each couple of types  $(i,j)$, we have $\alpha_{i,j}= 0$ or $\alpha_{i,j}= r$.
If, on all the machines, for each pair $(i,j)$ of tasks on the same machine we have $\alpha_{i,j}= 0$, then the maximum cost of a task is $1$ (its own size, $1$, times $\alpha_{i,i}=1$). Otherwise, there is a machine with two tasks of types $i$ and $j$ with  $\alpha_{i,j}= r$. The maximum cost is thus at least the cost of task $i$, which is at least $1\times \alpha_{i,i} + 1\times \alpha_{i,j}=1+r$. 

We show that the solution for the instance of {\sc PIC} is ``yes'' if and only if there is a solution of cost 1 for (at least) one of the $|V|$ instances of \pb.
Assume first that there is a solution for {\sc PIC}:
the vertices of $G$ can be partitioned into $k\leq K$ disjoint sets $V_1, V_2, \dots, V_k$
such that, for $1\leq i \leq k$, the subgraph induced by $V_i$  is a complete graph.
We take the $k$-th \pb instance. % among the corresponding instances of \pb.
For each $i \in\{1,\dots,k\}$, we assign to machine $M_i$ the tasks corresponding to the vertices of $V_i$.
%(a task corresponds to a vertex if it has the same number).
Since all the tasks on the same machine correspond to a clique in $G$, their coefficients $\alpha_{i,j}$ are all 0 ($i\neq j$). The only cost of a task $i$ is its own size times $\alpha_{i,i}$, that is $1$. Thus, the cost of the optimal solution of the $k$-th instance of \pb is 1.

Likewise, assume that there is a solution of cost 1 for (at least) one of the $|V|$ instances of \pb (wlog, for the  $k$-th instance). Then there is a ``yes'' solution for {\sc Partition into cliques}: since the maximum cost for the instance of \pb is 1, it means that all the values $\alpha_{i,j}$ between tasks on the same machines are 0 (for $i\neq j$) and thus that corresponding vertices form a clique in $G$.
% : it is possible to partition the vertices of $G$ into $k\leq K$ vertices.
\end{proof}

%%%%%%%%%%%%%%%%%%%%%%%%%%%%%%%%%%%%%%%%%%%%%%%%%%%%%%%%%%%%%%%%%%%%%%%%%%%%%%%%%%%
%                    			 PTAS					  %
%%%%%%%%%%%%%%%%%%%%%%%%%%%%%%%%%%%%%%%%%%%%%%%%%%%%%%%%%%%%%%%%%%%%%%%%%%%%%%%%%%%
\section{Approximation for fixed number of types}
The inapproximability proof of the previous section means that we can develop constant-factor approximations only for \pb with a constant number of types (and constant coefficients). We show in this section two approximation algorithms. First, a PTAS running in time ${O}(n^{T (\gamma k)^2})$, and thus mostly of theoretical interest. % as even for small instances the complexity is at least $\mathcal{O}(n^{72})$; and 
Then we introduce a fast greedy approximation algorithm.

\subsection{A PTAS}\label{sec:ptas}

\begin{algorithm}[!tb]
  \scriptsize
  \SetAlCapFnt{\scriptsize}
  \SetKwInOut{Input}{input}
  \SetKwInOut{Output}{output}

  % \Input{$J$ -- a set of tasks $J$; $C$ -- requested max cost; $k, \gamma$ -- parameters}
  
  % \Output{$\emptyset$ if there is no schedule of max cost $C$; or $\sigma$ -- a schedule of cost at most $C(1 + 1 / k)$}

    $J' = \emptyset$\;
    \For(\tcp*[h]{round down long tasks}){$j \in J, p_j \geq C/(\gamma k)$}{
      $p_{j'} = p_j - ( p_j \mod C/(\gamma k)^2$ ) \;
      $J' = J' \cup \{j'\}$ \;
    }
\For(\tcp*[h]{glue short tasks to containers}){$t \in T$}{
      $W^{(t)}_s = \sum_{j \in J^{(t)}, p_j < C/(\gamma k)}p_j$ ; \tcp*[h]{load of small tasks of type $t$} \; 
      \While{$W^{(t)}_s > 0$}{
        $p_{j''} = \min(C/(\gamma k), W^{(t)}_s)$ \;
        $J' = J' \cup \{ j'' \}$ ; \tcp*[h]{ $j''$ is a new container} \;
        $W^{(t)}_s = W^{(t)}_s - p_{j''}$ \;
      }
    }
    \lFor{$t \in T$}{ remove from $J'$ $m$ containers of type $t$ }
    $\sigma^{'*}$ = partition of $J'$ by solving (by dynamic programming) $OPT(n_1^{'(1)}, \dots, n_{(\gamma k)^2}^{'(1)}, \dots , n_1^{'(T)}, \dots, n_{(\gamma k)^2}^{'(T)}) = 1 + \min_{s_1^{(1)}, \dots, s_{(\gamma k)^2}^{(T)} \in \mathcal{C}} OPT(n_1^{'(1)} - s_1^{(1)}, \dots, n_{(\gamma k)^2}^{'(T)} - s_{(\gamma k)^2}^{(T)})$\;
    \vspace{-.7em} % note to the editor: this negative vspace fixes the too conservative layout because of the large subscript in the formula above
    \lIf{$\sigma^{'*}$ requires more than $m$ machines}{return $\emptyset$}
    $\sigma = \sigma^{'*}$ \;
    \For(\tcp*[h]{add removed containers}){k=1 \KwTo $m$}{
      \lFor{k=1 \KwTo $T$}{$\sigma[k] = \sigma[k] \cup \{ C/(\gamma k) \}$}
    }

    \For(\tcp*[h]{replace containers by small tasks}){k=1 \KwTo $m$}{
      \For{$t \in T$}{
        $i =$ number of type $t$ containers in $\sigma[k]$ \;
        replace $i$ containers by tasks of total load $W$, $iC/(\gamma k) \leq W \leq (i+1)C/(\gamma k)$\;
        }
    }
    replace in $\sigma$ rounded long tasks with original long tasks \;

\caption{\scriptsize{A PTAS for \pb with constant $T$ and $\alpha$}}\label{alg:ptas}
\end{algorithm}

Our PTAS (Algorithm~\ref{alg:ptas}) has a similar structure to the PTAS for $P||C_{\max}$\cite{hochbaum1987using}: the two main differences are the treatment of short tasks (which we pack into containers, and not simply greedy schedule) and the sizing of long tasks.
Our PTAS works even if $\alpha_{i,i}\neq 1$, and $\alpha_{i,j}\neq \alpha_{j,i}$. 
The algorithm uses parameters: $C$, the requested maximum cost; $k$, an integer; and $\gamma = T \alpha_{\max} \Big( 2 + 1/(\min \alpha_{i,i}) \Big)  $ (we assume that $T$ and $\alpha_{i,j}$  are constants). Given $C$, the algorithm either returns a schedule of cost at most $C (1 + 1/k)$, or proves that a schedule of cost at most $C$ does not exist.

The algorithm starts by constructing an instance $I'$ which will form a lower bound for $C$ of the original instance $I$. 
The algorithm partitions tasks into two sets: long tasks of size at least $C / (\gamma k)$; and short tasks. 
Long tasks are rounded down to the nearest multiple of $C / (\gamma k)^2$. 
Short tasks of a single type are ``glued'' into \emph{container tasks} of sizes $C / (\gamma k)$, 
except the last container task which might be shorter (of size $W^{(t)}_s \mod (C / (\gamma k))$, 
where $W^{(t)}_s$ is the load of short tasks of type $t$, 
$W^{(t)}_s = \sum_{j \in J^{(t)}, p_j < C/(\gamma k)}p_j$).
Then, the algorithm reduces the load in container tasks by removing $m$ containers (the shortest one and $m-1$ others) of each type. (Note that if the total load of short tasks of type $t$ is smaller than $m C / (\gamma k)$, there are less than $m$ containers, and they are all removed in this step; later, when reconstructing schedule, the algorithm adds the same number of containers that were removed. We omit this detail from Algorithm~\ref{alg:ptas} to make the code more readable). The resulting instance $I'$ has at most as many tasks and at most as high overall load as the original instance $I$ (the number of tasks and the load does not change only if all the tasks are long and their sizes are multiples of $C/(\gamma k)^2$). %%%%%%%%%%%%%%

The algorithm then schedules the lower-bound instance $I'$ using dynamic programming. For a given configuration $n_1^{'(1)}, \dots, n_{(\gamma k)^2}^{'(1)}, \dots , n_1^{'(T)}, \dots, n_{(\gamma k)^2}^{'(T)}$, where $n_i^{'(t)}$ is the number of tasks in $I'$ of type $t$ and size $i C / (\gamma k)^2$, $OPT$ denotes the minimal number of machines needed to schedule the configuration with cost smaller than $C$. To find $OPT$, the dynamic programming approach checks all possible configurations $\mathcal{C}$ of task sizes for a single machine $s_1^{(1)}, \dots, s_{(\gamma k)^2}^{(T)}$ (where $s_i^{(t)}$ denotes the number of tasks) that result in cost smaller than $C$, i.e.:
% \begin{align*}
$  s_1^{(1)}, \dots, s_{(\gamma k)^2}^{(T)} \in \mathcal{C} \Leftrightarrow \forall t \ \text{such that} \sum_i s_i^{(t)} >0:
  \sum_{t'} \sum_{i=1}^{(\gamma k)^2} \alpha_{t',t} s_i^{(t')} iC/(\gamma k)^2  \leq C \text{.}$
% \end{align*}
If OPT is larger than $m$, the algorithm ends. 
Otherwise, the returned schedule $\sigma^{'*}$ forms a scaffold to build a schedule $\sigma$ for the original instance $I$. 
First, the algorithm adds a container for each type on each machine (this container was removed before the dynamic programming).
Then, the algorithm replaces containers by actual short tasks. 
Assume that $\sigma^{'*}$ scheduled $i-1$ containers of type $t$ on machine $m$; 
the previous step added at most one container. 
The algorithm replaces $i$ containers of a total load $iC/(\gamma k)$ by scheduling unscheduled short tasks of type $t$ with a total load of at least $iC/(\gamma k)$ and at most $(i+1) C/(\gamma k)$ (which is always possible as a short task is shorter than $C/(\gamma k)$).
Finally, the algorithm replaces long tasks that were rounded down by the original long tasks.

\begin{proposition}
The PTAS returns a solution to \pb if and only if there is a solution of \pb of cost at most $C$. Moreover, if such a solution of cost $C$ exists, the cost of the solution returned by the PTAS is at most $C(1 + 1/k)$.
%\fixme{is this strong enough? should I also show that if PTAS returns an empty schedule, there is no schedule of cost $C$?}
\end{proposition}
\ifarxiv
\else
Proofs omitted due to space constraints are in the accompanying technical report~\cite{pascual2017maxcost-techreport}.
\fi
\begin{hideproof}
Assume first that there is an optimal schedule $\sigma^*$ of instance $I$ using $m$ machines and having cost at most $C$. Consider a schedule $\sigma'$ for $I'$ constructed according to $\sigma^*$. Each long task in $\sigma'$ is placed on the same machine as in $\sigma^*$. If $\sigma^*$ executes on a machine a total load $W^{*(t)}_s$ of small tasks of type $t$, this load is replaced by $\lfloor W^{*(t)}_s / (C/(\gamma k) \rfloor$ containers, each of size $C/(\gamma k)$. $\sigma'$ is a valid schedule for $I'$ as it schedules all tasks in $I'$. Moreover, the cost of $\sigma'$ is at most $C$, as the load of each type on each machine is not higher than the corresponding load in $\sigma^*$. As the dynamic programming used in PTAS analyses all possible schedules, it will return a schedule $\sigma^{'*}$ using at most the same number of machines as in $\sigma^*$.

Assume now that the dynamic programming returns a schedule $\sigma^{'*}$ of cost at most $C$. 
By adding a single container on each machine and each type, 
the cost increases by at most $T \alpha_{\max} C/(\gamma k)$. 
By replacing the containers by small tasks, 
the load of each type is increased by at most $C/(\gamma k)$, 
thus the cost is increased by at most $T \alpha_{\max} C/(\gamma k)$. 
Finally, by replacing the rounded-down long tasks by the tasks of the original sizes, 
the size of each long task is increased by at most $C / (\gamma k)^2$. 
As the cost of $\sigma^{'*}$ was at most $C$, and a long task is of size at least $C/(\gamma k)$, 
there are at most $\gamma k / \alpha_{i,i}$ tasks of each type $i$ (as the cost of each type on itself has to be smaller than $C$).
There are thus at most $T \gamma k / (\min \alpha_{i,i})$ long tasks in total. 
The total increase of cost due to long tasks (and the maximal influence between types)  is thus at most $\Big( T \gamma k / (\min \alpha_{i,i}) \Big) \alpha_{\max} \Big( C/(\gamma k)^2 \Big)$.
Consequently, the cost of $\sigma$ is bounded by $C(\sigma) \leq C + \Big( C T \alpha_{\max} / (\gamma k) \Big) \Big(2 + 1 / (\min \alpha_{i,i}) \Big) = C (1 + 1/k)$ (as $\gamma =  T \alpha_{\max} \Big( 2 + 1  / (\min \alpha_{i,i}) \Big)$).
\end{hideproof}

\begin{proposition}
The PTAS runs in time $\mathcal{O}(n^{T (\gamma k)^2})$.
%The PTAS runs in time $\mathcal{O}(mnT + n^{T (\gamma k)^2})$.
\end{proposition}
\begin{hideproof}
$T$ is a constant, and  we can assume than $n\geq m$ since otherwise the optimal solution is trivial: each task goes on a different machine. Thus, the runtime of all loops in Algorithm~\ref{alg:ptas} is bounded by $\mathcal{O}(n)$.  
We upper-bound the cost of the dynamic programming by computing the number of valid entries of the $(n_1^{'(1)}, \dots, n_{(\gamma k)^2}^{'(1)}, \dots , n_1^{'(T)}, \dots, n_{(\gamma k)^2}^{'(T)})$ vector.
This vector has $(T (\gamma k)^2)$~dimensions, and $n_i^{'(t)}\leq n$ for each pair $(i,t)$. Thus there are at most $n^{T (\gamma k)^2}$ distinct vectors. 
For each of these vectors, the algorithm must check at most as many entries as there are possible single machine configurations.
As there are at most $\gamma k / \alpha_{i,i}$ tasks of a type $i$ on a single machine,
there are at most $\Big( 1 + \gamma k/(\min \alpha_{i,i}) \Big)^{T (\gamma k)^2}$ such configurations to check. As $\gamma$, $\alpha_{i,j}$, $T$ and $k$ are constant, the number of configurations to check is a constant; thus, the complexity of the dynamic programming algorithm is dominated by the number of valid entries to check, $\mathcal{O} (n^{T (\gamma k)^2})$.
\end{hideproof}

%%%%%%%%%%%%%%%%%%%%%%%%%%%  Greedy  %%%%%%%%%%%%%%%%%%%%%%%%%%%%%%%%%%%%%%%%%%%%%%%%%%%%%%
\subsection{A Greedy list-scheduling approximation}\label{sec:greedy-list}

\greedy is a greedy $\frac{2 T m}{m-T}$-approximate algorithm for \pb with constant number of types.
%, which can be used when the PTAS is too slow (when $T$ is large in particular).
\greedy groups tasks by \emph{clusters}. All the tasks of the same type are in the same cluster. Two tasks of type $i$ and $j$ are in the same cluster iff their types are compatible ($\alpha_{i,j}\leq 1$ and $\alpha_{j,i}\leq 1$). While minimizing the number of clusters is NP-hard (by an immediate reduction from {\sc Partition into cliques}), any heuristics can be used, as the approximation ratio does not depend on the number of clusters.

%\fixme{Krz to Fanny: what do we do if $\alpha_{i,i} > 1$? F: indeed, we assume that $\alpha_{i,i}\leq 1$... We could say at the beginning of the paper that we assume that $\alpha_{i,i}=1$, and say un the paper that in fact the PTAS works for a very general case, even when $\alpha_{i,i}\neq 1$ and $\alpha_{i,j}\neq \alpha_{j,i}$?} OK done, moved to the model
% For example the greedy $O(T^2)$ algorithm which considers the types in an arbitrary order and group with the current examed type $i$ the not yet assigned types $j$ which are compatibles (such that  $\alpha_{i,j}\leq 1$ and $\alpha_{j,i}\leq 1$).
%Let $K$ be the number of constructed clusters.
% Let $W=\sum_{i=1}^n p_i$ be the total load.

% \fixme{Krz to Fanny: if $\alpha_{i,i} < 1$, $L$ is not anymore a lower bound, right? (F: yes, indeed) we don't need this here, but if we do bin search for $L$, $L_{min}$ should be $\frac{W}{(m-K)\alpha_{min}}$} % OK, now alpha_ii = 1 in the model
% Let $p_{\max}= \max_{1\leq i\leq n} p_i$ be the maximum size of a task.

Clusters are processed one by one. Each cluster is allocated to at least one, dedicated machine. (We assume that $m$, the number of machines, is smaller than or equal to the number of clusters $K$; $K \leq T$, and in a data center $T$ should be much smaller than $m$).
The algorithm puts tasks from a cluster on a machine until machine load reaches \allowbreak $L_{\max}=\max\{2L,L+p_{\max}\}$ (where $L=(\sum p_i)/(m-T)$ is the average load), then opens the next machine. In practice, rather than fixing the maximum machine load to $L_{\max}$, we do a dichotomic search over $[1, L_{\max}]$ to find the smallest possible threshold leading to a feasible schedule.
The complexity of \greedy with dichotomic search is $\mathcal{O}(T^2 n \log(L_{\max}))$.

%:  in the graph where each node is a type, and two types $i$ and $j$ are linked by an edge if $\alpha_{i,j}$

%camo
% \begin{algorithm}[tb]
%   \footnotesize
%   \SetAlCapFnt{\footnotesize}
%   \SetKwInOut{Input}{input}
%   \SetKwInOut{Output}{output}
%   \Input{$\mathcal{K}$, a set of clusters partitioning $J$}
%   %, partitioned into $K$ clusters $\{1, \dots, K\}$.% (two tasks of type $i$ and $j$ are in the same cluster only if $\alpha_{i,j}\leq 1$ and $\alpha_{j,i}\leq 1$).
% \Output{A $\frac{2Tm}{m-T}$-approximate solution.}
% $i=1; l=0; L=W/(m-T); Lmax=\max\{ 2L, L+p_{max} \}$ \;
% \ForEach(\tcp*[h]{allocate the tasks of cluster $k$}){cluster $k$}
% {
% \ForEach{task $j$ in cluster $k$ in an arbitrary order}
% {\If{$p_j + l > Lmax$}{
% $i=i+1; l = 0$\;
% }
% Assign $j$ to machine $i$; $l=l+p_j$\;
% }
% $i=i+1; l=0$\; 
% }
% \caption{\greedy : a greedy algorithm for \pb}\label{alg:greedy}
% \end{algorithm}

%\fixme{give a name to this algo ?}

\begin{proposition}\label{prop:fillgreedy-approx}
Algorithm \greedy is a  $\frac{2 T m}{m-T}$-approximate algorithm for \pb. 
\end{proposition}

\begin{proof}
We first show that the allocation is feasible, i.e. the algorithm uses at most $m$ machines.
Let $m_{used}$ be the number of machines to which at least one task is allocated.
Among these $m_{used}$ machines,
at most $K$ have load smaller than $L$.
Indeed, for each cluster the algorithm allocates tasks to a machine beyond $L$ (as $L_{\max} \geq L + p_{max}$),
unless there are no remaining tasks.
Thus, for each cluster, only the load of the last opened machine can be smaller than $L$.
Thus, the load allocated on these $m_{used}$ machines is at least $(m_{used}-K)L=(m_{used}-K)\frac{W}{m-T}$. Since the total load is $W$, we have $(m_{used}-K)\frac{W}{m-T}\leq W$. Thus $\frac{m_{used}-K}{m-T}\leq 1$, and so $m_{used}-K\leq m-T$. Since $K\leq T$, we have $m_{used}\leq m$. Thus, the allocation returned by \greedy is feasible.

We now show that the cost is $\frac{2 K m}{m-T}$-approximate.
We consider an instance $I$ of \pb. Let $\mathcal{O}$ be an optimal solution of $I$ for \pb, and let $OPT$ be the maximum cost of a task in $\mathcal{O}$. 
Since, for each type $i$, $\alpha_{i,i}=1$, we have $OPT\geq p_{max}$.
Let $L_{max}(\mathcal{O})$ be the maximum load of a machine in $\mathcal{O}$.
Let us consider that this load is achieved on machine $i$.
We have $L_{max}(\mathcal{O})\geq \frac{W}{m}$ (by the surface argument).
Since there are at most $T$ types on machine $i$, there is at least one type which has a load of at least $\frac{L_{max}(\mathcal{O})}{T}$ on machine $i$.
The cost of a task of this type on machine $i$ is thus at least $\frac{L_{max}(\mathcal{O})}{T}$, and therefore $OPT\geq \frac{L_{max}(\mathcal{O})}{T}\geq \frac{W}{Tm}$.

Let $\mathcal{S}$ be the solution returned by \greedy for instance $I$. Let $C(\mathcal{S})$ be the maximum cost of a task in $\mathcal{S}$. Let $L_{max}(\mathcal{S})$ be the maximum load of a machine in $\mathcal{S}$. Since two tasks $i$ and $j$ are scheduled on the same machine only if they belong to the same cluster, i.e. only if $\alpha_{t_i,t_j}\leq 1$, the cost of each task is at most equal to $L_{max}(\mathcal{S})$, and thus $C(\mathcal{S})\leq L_{max}(\mathcal{S})$. Moreover, by construction, we have $L_{max}(\mathcal{S}) \leq \max\{2L,L+p_{max}\}$. We consider the two following cases:
\begin{itemize}[nosep,leftmargin=.1em,labelwidth=*,align=left]
\item case 1: $\max\{L,p_{max}\}=p_{max}$. In this case, $C(\mathcal{S})\leq L_{max}\leq L+ p_{max}= \frac{W}{m-T} + p_{max} = \left( \frac{Tm}{m-T} \right) \frac{W}{Tm} + p_{max}$. Since  $OPT\geq p_{max}$ and $OPT\geq \frac{W}{Tm}$, we have $C(\mathcal{S})\leq (\frac{Tm}{m-T}+1)OPT < \frac{2 T m}{m-T} OPT$.
\item case 2: $\max\{L,p_{max}\}=L$. In this case, $C(\mathcal{S})\leq L_{max}\leq 2L = \frac{2 W}{m-T} \leq 2 \left( \frac{Tm}{m-T} \right) \frac{W}{Tm} \leq \frac{2 T m}{m-T} OPT$ because  $OPT\geq \frac{W}{Tm}$.
\end{itemize}
% \vspace{-1.4em}
%In both cases, the solution returned by \greedy is a $\frac{2 T m}{m-T}$-approximate solution. 
\end{proof}

\section{Heuristics}\label{sec:heuristics}

We propose a few other algorithms for \pb. These algorithms are fast approximations when $T=2$
\ifarxiv
(see Section~\ref{sec:two-types})
\else
(see~\cite{pascual2017maxcost-techreport}).
\fi
They all use as a subprocedure an algorithm $\mathcal{A}$ for $P||C_{\max}$, such as LPT (used in our experiments). $\mathcal{A}$ uses task's size $p_i$ as task's length.

{\sc SchedMixed} uses $\mathcal{A}$ on all tasks and all machines. Let $\sigma$ be the schedule constructed by $\mathcal{A}$ on $m$ machines with tasks $J$. {\sc SchedMixed}($\mathcal{A}$) returns the partition $P$ of the tasks equal to allocation in $\sigma$ (tasks on $M_i$ in $P$ are the tasks on $M_i$ in $\sigma$).

{\sc SchedJuxtapose} uses $\mathcal{A}$ on all machines for each type separately and then joins the schedules.
Let $\sigma_t$ be the schedule obtained by applying $\mathcal{A}$ on tasks $J^{(t)}$ of type $t$ on $m$ machines.
{\sc SchedJuxtapose} merges schedules reversing the order of machines for every other type, i.e.:
tasks on machine $M_i$ are tasks allocated to $M_i$ in $\sigma_{2k+1}$ and to $M_{m-i+1}$ in $\sigma_{2k}$
(when $\mathcal{A}=LPT$ and with a small number of tasks, the machines with smallest indices have the highest load).

{\sc BestSchedule}$(\mathcal{A}$) returns the partition with the lowest cost among the results of {\sc SchedJuxtapose}($\mathcal{A}$) and {\sc SchedMixed}($\mathcal{A}$).

{\sc GreedyDedicated}$(\mathcal{B}$) separates types into $K$ clusters (as in Section~\ref{sec:greedy-list}).
Clusters do not share machines.
%Each cluster $k$ will use $m_k$ distinct machines ($\sum_{k \in K} m_k=m$)
The algorithm runs a subprocedure $\mathcal{B}$ ({\sc SchedMixed}, {\sc SchedJuxtapose} or {\sc BestSchedule}) to put tasks of $k$-th cluster onto $m_k$ machines.
%partitions cluster $k$ load onto $m_k$ machines () using.
{\sc GreedyDedicated} returns the allocation with the minimal cost over all possibilities of assigning $[ m_k ]$ to clusters (by exhaustive search over $[m_k]: \sum_{k \in K} m_k=m$).

%The complexities of {\sc SchedMixed}, {\sc SchedJuxtapose} and {\sc BestSchedule} are the same as $\mathcal{A}$; {\sc GreedyDedicated}, because of the exhaustive search, is exponential in $T$ times the complexity of $\mathcal{A}$.

Let $C_\mathcal{A}$ be the complexity of Algorithm $\mathcal{A}$. Algorithm {\sc SchedMixed} is in $O(C_\mathcal{A})$;  {\sc SchedJuxtapose} and {\sc BestSchedule} are in $O(T C_\mathcal{A})$; {\sc GreedyDedicated} is in $O(K m^K C_\mathcal{A})$.

\ifarxiv
\section{Two Types}\label{sec:two-types}
To show the influence of the coefficient $\alpha$ on the shape of the optimal schedule, we study in this section a series of special cases with only two types $T=2$ and a symmetric coefficient $\alpha=\alpha_{t',t}=\alpha_{t,t'}$.
%We assume here that $\alpha_{i,i}=1$. \fixme{move to model}
We distinguish three cases based on coefficient $\alpha$: \emph{compatible} types when $\alpha \leq 1$; \emph{incompatible} when $1 < \alpha < 2$; and \emph{clashing} when $\alpha \geq 2$.\fixme{F: I changed the categorie of $\alpha=2$ since when $\alpha=2$ there is at most one shared machine - the example given with the $m$ tasks of length $p$ and with the small tasks of length $\varepsilon$ does not work in this case; on the contrary the proof of prop 4.6 still works when $\alpha=2$}\fixme{F: I also changed the divisible load example, so that it works for all the divisible load instances, and not only when the loads of type 1 and 2 are equals.}

To illustrate the difference between \emph{compatible} and \emph{incompatible} types, we consider divisible loads instances. The definition follows the one used in scheduling: if $W_i$ is the load of the tasks of type $i\in\{1,2\}$, this load can be assigned in any way to the machines (as if it was composed of a huge number of tiny tasks of size $\varepsilon$).

When $\alpha < 1$, the types are \emph{compatible}. 
A task from a different type results in a smaller cost than a task from the same type. 
Thus, in the optimal schedule for $I$ all machines are shared between the two  types: each machine executes a load $W_i/m$ of type $i$.  
When $\alpha=1$, any schedule with load $(W_1+W_2)/m$ on each machine is optimal (the impact of a task of type 1 or 2 on another task is the same).  
%$n/(2m)$ tasks of each type (in the boundary case of $\alpha=1$, the costs are equal, thus there is another class of optimal schedules as in the next case).

When $\alpha > 1$, a task from a different type results in a larger cost than a task from the same type. Thus, in a divisible load instance, there is at most one machine that is shared between the two types (if there were more than one shared machine, we could reduce the cost by exchanging tasks between the machines).

To distinguish the cases of \emph{incompatible} and \emph{clashing} types,
we now focus on non divisible load instances. When $\alpha<2$, there can be up to $m$ machines which have to be shared in an optimal solution, as in the following instance. There are $m$ tasks of each type: 
type 1 has only long tasks (of length $p$); 
type 2 has only short tasks (of length $\varepsilon$).
In a schedule with $m$ shared machines, the maximum cost is $\alpha p + \varepsilon$.
In a schedule with less than $m$ shared machines, 
at least one machine executes two tasks of type 1,
so the maximum cost is at least $2 p$. 
Thus, if $\alpha < 2$, and if $\varepsilon$ is sufficiently small, a schedule with $m$ shared machines has a lower cost. On the contrary, we prove in Proposition~\ref{prop:clashing-1mach} that if $\alpha \geq 2$ there is always at most a single shared machine. 
% \fixme{the proof is for $\alpha > 2$, not $\geq$}. % OK, fixed

%%%%%%%%%%%%%%%%%%%%%%%%%%%%%%%%%%%%%  Compatible types ($\alpha\leq 1$) %%%%%%%%%%%%%%
\subsection{Compatible types ($\alpha\leq 1$)}\label{sec:two-types-compatible}
% We show two approximation algorithms solving \pb by using an algorithm $\mathcal{A}$ for $P||C_{\max}$ (for instance, PTAS or LPT). We then propose an algorithm running both of them and returning the better partition.
%; and finally a fast approximation using LPT as $\mathcal{A}$.

%Let $\mathcal{A}$ be a (possibly approximate) algorithm for problem $(P||C_{max})$. We derive from $\mathcal{A}$ an approximate algorithm, denoted by

% {\sc SchedJuxtapose}($\mathcal{A}$) uses $\mathcal{A}$ on all machines for each type separately and then joins the schedules.
% Let $S_t$ be the schedule obtained by applying $\mathcal{A}$ on tasks $J^{(t)}$ of type $t$ on $m$ machines (here and in the remainder, $\mathcal{A}$ uses task's size $p_i$ as its length).
% {\sc SchedJuxtapose} merges the two schedules $S_1$ and $S_2$: tasks on machine $M_i$ are tasks allocated to $M_i$ in $S_1$ and in $S_2$.

%%%%%%%% Lemma 1
\begin{proposition}
\label{lemma1}
Let $\mathcal{A}$ be a $O(X)$, $(1+\varepsilon)$-approximate algorithm for $P||C_{max}$.  Algorithm {\sc SchedJuxtapose}($\mathcal{A}$) is a $O(X)$,  $(1+\varepsilon)(1+\alpha)$-approximate algorithm for \pb for $T=2$ and $\alpha \leq 1$. 
\end{proposition}

\begin{hideproof}
Let $I$ be an instance of \pb. Let $P$ be the partition returned by algorithm {\sc SchedJuxtapose}($\mathcal{A}$) for instance $I$. Let $Cost(P)$ be the cost of $P$.  
Let $P^*$ be an optimal solution of instance $I$ for \pb, and let $OPT$ be the cost of $P^*$.   Let $C^*_1$ (resp. $C^*_2$) be the makespan of an optimal schedule of $I_1$ (resp. $I_2$) for problem $(P||C_{max})$. Let $C_1$ (resp. $C_2$) be the makespan of schedule $S_1$ (resp. $S_2$). We have: $C_1\leq (1+\varepsilon) C^*_1$ and  $C_2\leq (1+\varepsilon) C^*_2$ since $\mathcal{A}$ is a  $(1+\varepsilon)$-approximate algorithm for problem $(P||C_{max})$.  Moreover, we have $C^*_1\leq OPT$ since in $P^*$ all the tasks of $I_1$ are partitioned into at most $m$ subsets (machines): the maximum load of tasks of type 1 on a same machine in $P^*$ ois at least $C^*_1$. Since the cost of a task of type $1$ is at least the load of the tasks of type 1 on the same machine (because $\alpha_{1,1}=1$), we have $OPT\geq C^*_1$. 
Likewise, we have $OPT\geq C^*_2$. 
Let us assume without loss of generality that $C_1\leq C_2$. The cost of $P$ is smaller than or equal to $C_1+\alpha C_2$, since the cost of a task of type 1 is the load of the tasks of type 1 on the same machine (at most $C_1$) plus $\alpha\leq 1$ times the load of the tasks of type 2 on the same machine (at most $C_2$).  Thus $Cost(P)\leq C_1+  \alpha C_2 \leq (1+\varepsilon)( C^*_1 +  \alpha  C^*_2) \leq (1+\varepsilon)(1+\alpha) OPT$.
\end{hideproof}
 
% The second algorithm, {\sc SchedMixed}, uses $\mathcal{A}$ on all tasks and all machines. Let $S$ be the schedule constructed by $\mathcal{A}$ on $m$ machines with tasks $J=J^{(1)} \cup J^{(2)}$. Algorithm {\sc SchedMixed}($\mathcal{A}$) returns the partition $P$ of the tasks equal to allocation in $S$ (i.e. the tasks on $M_i$ in $P$ are the tasks on $M_i$ in $S$).

%%%%%%%% Lemma 2 
\begin{proposition}
\label{lemma2}
Let $\mathcal{A}$ be a  $(1+\varepsilon)$-approximate algorithm for $P||C_{max}$.  {\sc SchedMixed}($\mathcal{A}$) is a  $\frac{2(1+\varepsilon)}{1+\alpha}$-approximate algorithm for \pb for $T=2$ and $\alpha \leq 1$. 
\end{proposition}

\begin{hideproof}
Let $I$ be an instance of \pb. Let $P$ be the partition returned by algorithm {\sc SchedMixed}($\mathcal{A}$) for instance $I$. Let $Cost(P)$ be the cost of $P$.  %Let $I_1$ (resp. $I_2$) be the set of tasks of type 1 (resp. type 2) of $I$.   Let $C^*_1$ (resp. $C^*_2$) be the makespan of an optimal schedule of the tasks of  $I_1$ (resp. $I_2$) on $m$ machines for problem $(P||C_{max})$. 

Let $C_{max}^*$ be the makespan of an optimal solution of problem $(P||C_{max})$ on instance $I$. Let $C_{max}$ be the makespan of the schedule returned by $\mathcal{A}$  on instance $I$. Since $\mathcal{A}$  is a $(1+\varepsilon)$-approximate algorithm for problem $(P||C_{max})$, we have $C_{max}\leq (1+\varepsilon) C_{max}^*$. Moreover, $Cost(P)\leq C_{max}$ since the cost of each task is equal to the load of the tasks of the same type on the same machine times $\alpha$ times  the load of the tasks of the same type on the same machine, and $\alpha\leq 1$. 

Let $P^*$ be an optimal solution of instance $I$ for \pb, and let $OPT$ be the cost of $P^*$.    Let $M_i$ be the most loaded machine in $P^*$ and let $C_{max}(P^*)$ be the makespan of $P^*$ (i.e.  $C_{max}(P^*)$ is equal to the sum of the sizes of the tasks on $M_i$ in $P^*$). Let $L_1$ (resp. $L_2$) be the load of the tasks of type 1 (resp. type 2) on $M_i$ in $P^*$. Wlog, assume that $L_1\geq L_2$. The cost of the tasks of type 1 on $M_i$ is $L_1+\alpha L_2$. Thus, 
$OPT\geq L_1+\alpha L_2=L_1+\alpha(C_{max}(P^*)-L_1)\geq \frac{C_{max}(P^*)}{2}+\alpha  \frac{C_{max}(P^*)}{2}= (\frac{1+\alpha}{2})  C_{max}(P^*)$. The last inequality holds because $L_1\geq  \frac{C_{max}(P^*)}{2}$ and $\alpha\leq 1$. 

Since $Cost(P)\leq C_{max}\leq  (1+\varepsilon) C_{max}^*$ and $OPT\geq (\frac{1+\alpha}{2})   C_{max}(P^*) \geq  (\frac{1+\alpha}{2})   C_{max}^*$, we have $Cost(P)\leq \frac{2(1+\varepsilon)}{1+\alpha} OPT$.
\end{hideproof}
 
  % \medskip
  
{\sc SchedJuxtapose} has lowest approximation for $\alpha$ close to 0, while {\sc SchedMixed} has lowest approximation for $\alpha$ close to 1.
% We now combine the two algorithms into {\sc BestSchedule}$(\mathcal{A}$), which 
% returns the partition with the lowest cost among the result of {\sc SchedJuxtapose}($\mathcal{A}$) and {\sc SchedMixed}($\mathcal{A}$).

%%%%%%%% Prop
\begin{proposition}
\label{prop:racine2}
Let $\mathcal{A}$ be a  $(1+\varepsilon)$-approximate algorithm for problem $(P||C_{max})$, which runs in $O(X)$.  Algorithm {\sc BestSchedule}($\mathcal{A}$) is a $O(X)$,   $\sqrt(2)(1+\alpha)$-approximate algorithm for \pb for $T=2$ and $\alpha \leq 1$.
\end{proposition}

\begin{hideproof}
The approximation ratio of {\sc SchedJuxtapose}($\mathcal{A}$) is  $(1+\varepsilon)(1+\alpha)$ (cf. Proposition \ref{lemma2}). The approximation ratio of {\sc SchedMixed}($\mathcal{A}$) is $\frac{2(1+\varepsilon)}{1+\alpha}$ (cf. Proposition \ref{lemma1}). Thus, the approximation ratio of algorithm {\sc BestSchedule}($\mathcal{A}$) is $\min \{(1+\varepsilon)(1+\alpha), \frac{2(1+\varepsilon)}{1+\alpha}   \}$. 
The maximum is achieved when $(1+\varepsilon)(1+\alpha) = \frac{2(1+\varepsilon)}{1+\alpha}$
% \begin{equation*}
%\label{eqFL:r}
% \begin{aligned}
% &(1+\varepsilon)(1+\alpha) = \frac{2(1+\varepsilon)}{1+\alpha}\\
$\Rightarrow (1+\alpha) = \frac{2}{1+\alpha}
\Rightarrow (1+\alpha)^2=2
\Rightarrow \alpha^2 + 2\alpha -1 =0$
%C \left( (y_t^j) ,r\right) &= \min_{ (x_t^j: x_t^j \in \{0,\dots,y_t^j\}) }  
%\Biggl( C\left( (y_t^j - x_t^j), r-1\right) \\
% &+ 
%C\left( (x_t^j), 1\right)
%\Biggr). 
% \end{aligned}
% \end{equation*}
Since $\alpha\geq 0$, this means that $\alpha = \frac{-2+\sqrt{8}}{2} =\sqrt{2}-1$. The maximum approximation ratio is obtained for $\alpha=\sqrt{2}-1$. It is thus $(1+\varepsilon)(1+\alpha)=\sqrt{2}(1+\varepsilon)$.Therefore, {\sc BestSchedule}($\mathcal{A}$) is $\sqrt{2}(1+\varepsilon)$-approximate. 
Furthermore, if $\mathcal{A}$ runs in $O(X)$, then {\sc SchedJuxtapose}($\mathcal{A}$) also runs in $O(X)$ since it runs twice algorithm $\mathcal{A}$, and algorithm {\sc SchedMixed}($\mathcal{A}$) also runs in $O(X)$ since it simply runs once algorithm $\mathcal{A}$. Therefore, algorithm {\sc BestSchedule}($\mathcal{A}$), which runs once  {\sc SchedJuxtapose}($\mathcal{A}$) and once {\sc SchedMixed}($\mathcal{A}$) also runs in $O(X)$. 
\end{hideproof}

\begin{corollary}
\label{prop:lpt}
Let ${LPT}$ be the algorithm which greedily schedules $n$ tasks from the longest one on parallel machines.  Algorithm {\sc BestSchedule}($LPT$) runs in $O(n \log n)$ and has an approximation ratio of   $\frac{4\sqrt{2}}{3}<1.89$ for \pb.
\end{corollary}

%This corollary follows from Proposition~\ref{prop:racine2}, using the fact that $LPT$ runs in $O(n \log n)$ and has an approximation ratio of $\frac{4}{3}$~\cite{graham1969bounds}.

%%%%%%%%%%%%%%%%%%%%%%%%%%%%%%%%%%%%  Incompatible types ($1 < \alpha < 2$) %%%%%%%%%%%%%%
\subsection{Incompatible types ($1 < \alpha < 2$)}\label{sec:incompatible}
In this section, we show that {\sc SchedMixed}
is a $\alpha(1+\varepsilon)$ approximation when $1 < \alpha \leq 2$. We also show \lptbytypes, a 2-approximate greedy algorithm running in $O(n)$.

%The first algorithm simply uses the $(1+\varepsilon)$-approximate algorithm for problem $(P||C_{max})$ on $m$ machines with $n$ tasks of lengths $p_1\dots, p_n$. It returns $\alpha(1+\varepsilon)$-approximate solutions for \pb.
%The second algorithm is a fast greedy algorithm which returns  2-approximate solution for \pb. 

%%%%%%%% Prop 1
\begin{proposition}
\label{prop:2-approx-incompatible}
Let $\mathcal{A}$ be a  $(1+\varepsilon)$-approximate algorithm for problem $(P||C_{\max})$, which runs in $O(X)$. This algorithm returns an $\alpha(1+\varepsilon)$-approximate solution for \pb in $O(X)$, for $T=2$ and $1<\alpha\leq2$. 
\end{proposition}

\begin{hideproof}
(sketch) Similar to the proof of Proposition~\ref{lemma2}: as machines are shared, the cost is within $\alpha$ of $C_{\max}$.
  % Let us consider an instance $I$ of \pb. This isntance can be considered as an instance of problem  $(P||C_{\max})$, by considering $m$ machines and tasks of  lengths $p_1\dots, p_n$. Let $C_{\max}$ be the makespan of the schedule returned by $\mathcal{A}$ on $I$ and let $C_{\max}^*$ be the optimal makespan of a schedule of $I$ for problem  $(P||C_{\max})$. We have $C_{\max}\leq (1+\varepsilon) C_{\max}^*$ since $\mathcal{A}$ is a $(1+\varepsilon)$-approximate algorithm for problem $(P||C_{\max})$. Let us now consider the schedule returned by $\mathcal{A}$  as a solution for \pb (the assignment of the tasks to the  machines is the one done in the schedule returned by $\mathcal{A}$). 
% The maximum load of a machine in this solution is, by definition, $C_{\max}$. Since $\alpha_{i,i}<\alpha$, the  cost of a task on a machine is smaller than $\alpha$ times the load of its machine, and thus smaller than $\alpha C_{\max}$. Likewise, since the maximum load of any solution of \pb on $I$ is at least $C_{\max}^*$, we have $OPT\geq C_{\max}^*$. Hence, the maximum cost of a task in the solution returned by $\mathcal{A}$ is at most $\alpha C_{\max}\leq \alpha (1+\varepsilon) C_{\max}^* \leq \alpha (1+\varepsilon) OPT$.
\end{hideproof}

%\fixme{F : We should perhaps define once and for all $W$ and $p_{max}$ at the beginning of the paper - we should check whether it is often used elsewhere}. -> OK, done

We now define a fast greedy algorithm for \pb called \lptbytypes. The algorithm is similar to \greedy, filling machines until a treshold and opening a new machine for the second type: the difference is that one machine can be shared between types.
\lptbytypes assigns first all the tasks of type 1, and then all the tasks of type 2.
%For each type, tasks are assigned in order of non-increasing sizes. 
Tasks gradually fill machines (as in \greedy): a task is assigned to the current machine if the resulting total load on that machine is at most $L=W/m + \max\{W/m, p_{\max}\}$; otherwise, a new machine is opened.
Also, a new machine is opened for the first task of the second type.
When the algorithm tries to open machine $m+1$, instead all the remaining tasks are assigned to the last machine the first type used.

\begin{proposition}
\label{prop:lpt-by-types}
Algorithm \lptbytypes is a $O(n)$,  2-approximate algorithm for \pb for $T=2$ and $1<\alpha \leq 2$.
\end{proposition}
\begin{hideproof}
% Let us first notice that the the tasks of type 2 cannot be assigned to $m_1$  before that all the tasks of type 1 have been scheduled (and thus before than $m_1$ has been defined). Indeed, assume by contradiction that this is the case (i.e. during the execution, we enter in the ``else" condition, whereas $i=1$). It means that on each machine, the load is larger than $\frac{W}{m}$, and thus that the total load assigned to the $m$ machines  is larger than $W$, which is impossible. So the tasks of type 2 are assigned to the machines once all the tasks of type 1 have been assigned, and the solution returned by \lptbytypes is thus feasible. 
% We Let us now show that this solution is 2-approximate. 
The only machine which might have tasks of both types is $M_1$.
% On all the machines, except possibly on machine $m_1$, 
% the tasks are either of type 1 or of type 2 (the only machine which can be shared by the two 
% types is machine $m_1$). 
Moreover, on all the machines except $M_1$, the load is at most $L=W/m+\max\{W/m, p_{\max}\}$.
Let $OPT$ be the cost of an optimal solution for \pb.
As in proof of Proposition~\ref{prop:fillgreedy-approx},
$OPT\geq p_{\max}$ and $OPT \geq W/m$.
% We have $OPT\geq p_{max}$ (because a task of size $p_{max}$ has a cost of at least $p_{max}$ since $\alpha_{i,i}=1$). Moreover, $OPT\geq \frac{W}{m}$ since the minimum maximum load is $\frac{W}{m}$, and the cost of a task is at least the load of the machine where it is (since $\alpha_{i,i}=1$ and $\alpha>1$).
Thus, the cost of each task executing on machine different than $M_1$ is at most $L<2 OPT$.
\\
%\medskip
We now show that the cost of a task assigned to $M_1$ is also at most $2OPT$.
If there are only tasks of type 1 on $M_1$,
then the load of this machine is at most $L$ (otherwise, the total allocated load of type 1 would be greater than W),
thus the cost of the tasks on $M_1$ is at most $L<2 OPT$.
If there are tasks of type 2 on $M_1$,
on all other machines the load is larger than $W/m$.
Thus the load on $M_1$ is smaller than $W-(m-1)W/m=W/m$.
Therefore, the cost of a task on $M_1$ is smaller than $\alpha W/m \leq 2 OPT$ since $\alpha\leq 2$.
Hence, the solution returned by \lptbytypes is 2-approximate.
% Let us assume that the load of the tasks of type 1 on $m_1$ is $x<L$. Therefore the load of 
%
%Let us first notice that the booléan variable SharedMachines cannot be set to true before that all the types of type 1 have been scheduled. Indeed, if this is the case, it would mean that on all the machines, the load is larger than $\frac{W}{m}$, and thus that the total load scheduled is larger than $W$, which is impossible. So SharedMachines will be set to true after that all the types of type 1 have been scheduled, and the schedule returned by \lptbytypes is thus feasible. 
%
%Let $C^*_{max}$ be the minimal maximum load of a machine in a schedule where tasks of sizes $p_1,\dots,p_n$ are scheduled on $m$ machines (i.e. $C^*_{max}$ is the makespan of an optimal solution of problem $(P||C_{max})$ where the tasks of length $p_i$ have to be scheduled on $m$ machines). We know that $C_{max}\geq \max\{W, p_{max}\}$. \fixme{F to F: to do}
\end{hideproof}

%%%%%%%%%%%%%%%%%%%%%%%%%%%%%%%%%%%%  Clashing types ($\alpha > 2$) %%%%%%%%%%%%%%%%%%%%%%
\subsection{Clashing types ($\alpha \geq 2$)}\label{sec:two-types-clashing}
For large coefficients, we show that an optimal solution uses at most one shared machine. We then use this result to show that an algorithm that uses no shared machines is a $(1 + \frac{1}{1+\alpha})(1+\varepsilon)$ approximation.

\begin{proposition}
If $\alpha \geq 2$, an optimal solution uses at most one shared machine.\label{prop:clashing-1mach}
\end{proposition}
\begin{hideproof}
The proof is by contradiction. Assume that in the optimal solution there are at least 2 shared machines, $k$ and $l$. 
Machine $k$ executes load of $a = W_k^{(1)}$ of type 1 and $b = W_k^{(2)}$ of type 2; 
machine $k'$ executes load of $a' = W_{k'}^{(1)}$ of type 1 and $b' = W_{k'}^{(2)}$ of type 2. 
Without loss of generality, we label types such that type 1 has higher load ($a + a' \geq b + b'$) and machines such that machine $k$ is allocated most of type 1 load ($a \geq a'$). The cost of a solution with two dedicated machines is $\max( a + a', b + b') = a + a'$.

Assume first that $a \geq b$; thus, cost of tasks on machine $k$ is at most $b + \alpha a$.
By contradiction, a solution with two shared machines is optimal, thus
$b + \alpha a \leq a + a'$. As $b > 0$,\fixme{I replaced $b\geq 0$ by $b>0$, since the machine is necessarily a shared machine - this will lead to Prop 4.6 true even if $\alpha=2$} $\alpha a < a + a'$. As $a \geq a'$, $\alpha a < 2 a$, which leads to a contradiction if $\alpha \geq 2$. 

In the second case, $a < b$. Thus, cost of tasks on machine $k$ is  $a + \alpha b$; and $a + \alpha b > a + \alpha a > \alpha a$. By contradiction, a solution with two shared machines is optimal, thus $\alpha a < a + a'$, which, as in the previous case, leads to a contradiction if $\alpha \geq 2$.
\end{hideproof}

Given an algorithm $\mathcal{A}$ for $P||C_{\max}$, the algorithm {\sc GreedyDedicated}, executes, for each possible value of $m_1 \in \{1,\dots,m-1\}$, algorithm $\mathcal{A}$ twice:
first, for tasks $J^{(1)}$ of type 1 scheduled on the first $m_1$ machines; 
second, for tasks $J^{(2)}$ of type 2 scheduled on the remaining $m-m_1$ machines. 
{\sc GreedyDedicated}, out of $(m-1)$ possibilities, returns the allocation with the smallest makespan (cost).
%This algorithm is an approximate algorithm for \pb
%(the proof, which uses Prop. \ref{prop:clashing-1mach}, is in the accompanying technical report.
%~\cite{maxcost16full}.

\begin{proposition}
\label{prop:greedy-clashing}
Given a $(1+\varepsilon)$-approximate algorithm $\mathcal{A}$ for $P||C_{\max}$ which runs in $O(X)$, {\sc GreedyDedicated} is a $O(m X)$, $(1 + \frac{1}{1 + \alpha}) (1 + \varepsilon)$-approximate algorithm for \pb for $T=2$ and $\alpha \geq  2$.
\end{proposition}
\begin{hideproof}
For $\alpha \geq 2$, there is at most one shared machine in the optimal allocation (Proposition~\ref{prop:clashing-1mach}).
Assume that the shared machine executes load $W^{(1)}$ of type 1 and $W^{(2)}$ of type 2.
The cost on the shared machine must be at most OPT, thus $W^{(1)} + \alpha W^{(2)} \leq OPT$ and $W^{(2)} + \alpha W^{(1)} \leq OPT$.
Thus, $W^{(1)} + W^{(2)} \leq \frac{2OPT}{1 + \alpha}$.
Construct now a schedule $\sigma'$ constructed from $\sigma^*$ that will not use a shared machine. The smaller out of $W^{(1)}$ and $W^{(2)}$ is moved to the first machine dedicated for its type. The cost of $\sigma'$ is at most $OPT (1 + \frac{1}{1+\alpha})$ and it is equal to the makespan on some machine.\\
%\medskip
Consider now an allocation returned by {\sc GreedyDedicated} and denote its makespan by $C_{\max}(\sigma(m_1^*))$. As $\alpha_{t,t}=1$ and no machine is shared, the cost is equal to the makespan $C_{\max}(\sigma(m_1^*))$.
As it is a $(1+\varepsilon)$-approximation of $C^*_{\max}$ 
and $C^*_{\max} \leq C_{\max}(\sigma')$, 
$C_{\max}(\sigma(m_1^*)) \leq (1+\varepsilon)C_{\max}(\sigma')$,
thus $C_{\max}(\sigma(m_1^*)) \leq  (1+\varepsilon) OPT (1 + \frac{1}{1+\alpha})$.
\end{hideproof}
\fi

\section{Experiments}\label{sec:experiments}
\ifarxiv
\subsection{Method}
\begin{table}[tb]
  \caption{Coefficients $\alpha_{t,t'}$ for 2, 3 and 4 types (rows) and 4 variants of setting the coefficients (columns).}\label{tab:alphas}
  \resizebox{\columnwidth}{!}{%
  {\footnotesize
    \begin{tabular}{l |l| l| l|l}
      &compatible & mixed & incompatible & clashing\\
\hline
$T=2$&
\tabcolsep=0.11cm\begin{tabular}{c c }
  1 & 0.5 \\%\hline
  0.5 & 1 \\
\end{tabular}
&
&
\tabcolsep=0.11cm\begin{tabular}{c c }
  1 & 1.5 \\%\hline
  1.5 & 1 \\
\end{tabular}
&
\tabcolsep=0.11cm\begin{tabular}{c   c }
  1 & 2 \\%\hline
  2 & 1 \\
\end{tabular}
    \\[4ex]
\hline
$T=3$&
\tabcolsep=0.11cm\begin{tabular}{c   c   c }
  1 & 0.5  & 0.25 \\%\hline
  0.5 & 1 & 0.5\\%\hline
  0.25 & 0.5 & 1 \\
\end{tabular}
&
\tabcolsep=0.11cm\begin{tabular}{c   c   c }
  1 & 0.5 & 1.5 \\%\hline
  0.5 & 1 & 1.5 \\%\hline
  1.5 & 1.5 & 1 \\
\end{tabular}
&
\tabcolsep=0.11cm\begin{tabular}{c   c   c }
  1 & 1.3 & 1.6 \\%\hline
  1.3 & 1 & 1.3 \\%\hline
  1.6 & 1.3 & 1 \\
\end{tabular}
&
\tabcolsep=0.11cm\begin{tabular}{c   c   c }
  1 & 2 & 3 \\%\hline
  2 & 1 & 2 \\%\hline
  3 & 2 & 1 \\
\end{tabular}
    \\[4ex]
\hline
$T=4$&
\tabcolsep=0.11cm\begin{tabular}{c  c  c  c}
  1 & 0.75  & 0.5  & 0.25 \\%\hline
  0.75 & 1 & 0.75 & 0.5 \\%\hline
  0.5 & 0.75 & 1 & 0.75 \\%\hline
  0.25 & 0.5 & 0.75 & 1
\end{tabular}
&
\tabcolsep=0.11cm\begin{tabular}{c   c   c   c}
  1 & 0.5 & 1.5 & 2 \\%\hline
  0.5 & 1 & 1.5 & 2 \\%\hline
  1.5 & 1.5 & 1 & 0.5 \\%\hline
  2 & 2 & 0.5 & 1
\end{tabular}
&
\tabcolsep=0.11cm\begin{tabular}{c   c   c   c}
  1 & 1.25 & 1.5 & 1.75 \\%\hline
  1.25 & 1 & 1.25 & 1.5 \\%\hline
  1.5 & 1.25 & 1 & 1.25 \\%\hline
  1.75 & 1.5 & 1.25 & 1
\end{tabular}
&
\tabcolsep=0.11cm\begin{tabular}{c   c   c   c}
  1 & 2 & 3 & 4 \\%\hline
  2 & 1 & 2 & 3 \\%\hline
  3 & 2 & 1 & 2 \\%\hline
  4 & 3 & 2 & 1
\end{tabular}
    \\
\hline
  \end{tabular}
}
}
%\fixme{indeed this tabular is nice and by using it we can decrease the paragraph devoted to the genration of $\alpha$ - we should just add one or two sentences explaining what is inside this tabular?}OK
\end{table}

\subsubsection{Data}
We used the Google Cluster Trace~\cite{reiss2012heterogeneity}, the standard dataset for datacenter/cloud resource management research, as an input data. The trace describes all tasks running during a month on one of the Google clusters. For each task, the trace reports in its task record table, among other data, the task's CPU, memory and disk IO usage averaged over a 5-minute long period.  This trace is certainly not ideal for our needs: the trace reports the usage of raw resources (CPU, memory, network, disk), and not the load of applications. However, to our best knowledge, there are no publicly-available traces describing loads and performance of applications (in contrast to raw resources). 

We generate a random sample of $10 000$ task records. Each task record corresponds to a task in our model. To generate loads and types, we use data on the mean CPU utilization and the assigned memory. We normalize CPU and memory utilization to their respective maximums. We remove 45\% of task records that reported less than $0.005$ in both normalized CPU and memory usage.

To assign one of $T$ types to a task, we analyze the ratio $\rho$ of the weighted CPU to the weighted memory usage. For $T=2$, a task with $\rho \leq 1$ is of type 1 (memory-intensive), and a task of $\rho > 1$ is of type 2 (CPU-intensive); this partitions the dataset in almost equal halves. For $T=3$, we pick tresholds $\log(\rho_1)=-0.66$ and $\log(\rho_2)=0.66$. 
(We chose these treshold from the histogram of the distribution of $\rho$ --- they correspond to values of $\rho$ for which the number of tasks significantly diminishes):
Type 1 is memory-intensive (10\%), type 3 is CPU-intensive (11\%) and type 2 is a mixed CPU-memory task.
(An alternative would be to use the disk IO measurements also reported in the trace; however, in only 3\% of tasks the normalized disk IO dominates both CPU and memory, which would result in 3-type instances having a very small number of tasks of the 3rd type, and thus virtually identical to 2-type instances.)

For $T=4$ we use tresholds of $\log(\rho_1)=-0.66$, then $\rho_2=1$, and $\log(\rho_3)=0.66$; type 1 is memory-intensive (10\%), type 2 is memory-CPU (40\%), 

To assign load to a task, we take the maximum from the weighted CPU and weighted memory, multiply this maximum by 100 and round to the nearest integer.

To generate an instance of $n$ tasks belonging to $T$ types, we take a random sample of $n$ tasks from the dataset; thus, the proportions of types in the generated instance are similar to the dataset. However, a random sample might have less than $T$ types: if it is the case, we remove a task from the most common type in the instance and add a task of the missing type.

We generate the coefficients $\alpha$ in four different ways. In all instances coefficients are symmetric ($\alpha_{t,t'}=\alpha_{t',t}$) and normalized ($\alpha_{t,t}=1$):
\begin{itemize}[nosep,leftmargin=.1em,labelwidth=*,align=left]
\item \emph{compatible}: smaller than 1;
\item \emph{incompatible}: between 1 and 2;
\item \emph{clashing}: at least 2;
\item \emph{mixed}: 2 incompatible clusters (see Section~\ref{sec:greedy-list}).
\end{itemize}
Table~\ref{tab:alphas} shows all coefficients: for instance, for $T=4$ and \emph{mixed} instances, $\alpha_{3,4}=0.5$.
In \emph{compatible} instances, the further apart the type numbers are, the lower the coefficients. In \emph{incompatible} instances, the further apart the type numbers are, the higher the coefficients; in \emph{clashing} instaces the values are higher than in \emph{incompatible}. Finally,  \emph{mixed} instances have both compatible and incompatible types: e.g., for three types, types 1 and 2 are compatible ($\alpha_{1,2}=0.5$), while both are incompatible with the third type

Note that the way we set the coefficients in non-compatible scenarios does not correspond to the way we partition the trace into types. We continue with these discretionary values as, first, we want to test our algorithms for variety of settings; and, second, we are not aware of any better dataset.

We generate instances with four different ways of setting the coefficients, the number of types $T \in \{ 2, 3, 4\}$ and two sizes: in \emph{small} instances, the number of tasks $n \in \{ 10, 20, 50 \}$ and the number of machines $m \in \{ 2, 3, 5, 10 \}$ (we generate all possibilities). In \emph{large} instances, the number of tasks $n \in \{ 200, 500, 1000 \}$ and the number of machines $m \in \{ 20, 50, 100 \}$ (again, we generate all possibilites). Finally, we discard unfeasible combinations: for \emph{clashing} and \emph{incompatible} coefficients, instances in which the number of types is higher than the number of machines; for \emph{mixed} instances, instances with the number of types smaller than 3. For each feasible combination, we generate 30 instances. Overall, we generate 6390 feasible instances.

\subsubsection{Algorithms}
We study the following algorithms:
\begin{itemize}[nosep,leftmargin=.1em,labelwidth=*,align=left]

\item \greedy, denoted by $fill$ in plots (Section~\ref{sec:greedy-list}): we use binary search to optimize the bound up to which each machine is loaded. We also sort tasks in each cluster by decreasing lengths (which makes our algorithm analogous to the last-fit decreasing bin packing algorithms).
\item {\sc SchedJuxtapose} ($jux$ in plots), (Section~\ref{sec:two-types-compatible}): When juxtaposing schedules of different types, we reverse the order of machines for every other type (as in LPT with a small number of tasks, the machines with smallest indices have the highest load).
\item{\sc SchedMixed} ($mix$ in plots), (Section~\ref{sec:two-types-compatible});
\item{\sc BestSchedule} ($best$ in plots), (Section~\ref{sec:two-types-compatible});
\item{\sc \lptbytypes}  ($g2$ in plots), (Section~\ref{sec:incompatible});
\item{\sc GreedyDedicated}  ($ded$ in plots), (Section~\ref{sec:two-types-clashing}).
\end{itemize}

We generalize algorithms from Section~\ref{sec:two-types} for multiple types. We also do binary search to minimize the boundary value in \greedy and \lptbytypes. We use LPT as $\mathcal{A}$, the underlying scheduling algorithm for the single type problem ($P||C_{\max}$). LPT orders tasks by decreasing sizes; tasks are processed sequentially and each task is assigned to a machine with the smallest total load (note that we consider here the load and not the cost).

We run \greedy on all variants of coefficients; {\sc SchedJuxtapose}, {\sc SchedMixed} and {\sc BestSchedule} on \emph{compatible} coefficients;
{\sc SchedMixed}, \lptbytypes and {\sc GreedyDedicated} on \emph{incompatible} instances;
and {\sc GreedyDedicated} on \emph{clashing} instances. 
On \emph{mixed} instances, we run {\sc SchedMixed} (as in Section~\ref{sec:incompatible}), and \lptbytypes (here the 2 types are the two clusters of compatible types). We also run {\sc GreedyDedicated} between the two clusters. In this case the algorithm used inside the clusters (algorithm $\mathcal{A}$) is either {\sc SchedJuxtapose} (denoted by $d-jux$ in plots), {\sc SchedMixed}  ($d-mix$ in plots) or {\sc BestSchedule} ($d-best$ in plots).

\subsubsection{Scoring}
We compared the maximum cost returned by the algorithms to the lower bound and computed the relative performance. We used the following lower bounds. (1) $p_{\max}$, the maximum size of the task (as in our experiments for each type $\alpha_{t,t}=1$, the cost on the machine on which the longest task is allocated is at least $p_{\max}$. (2) For incompatible and clashing instances, the average load of a machine, $W / m$. (3) For compatible instances, a solution of the following LP:
\begin{align*}
&\min c, \text{such that}\\[-.5em]
&\forall k \in [1, m], \forall t' :  \sum_{t} x_{t, k} W^{(t)} \min(1, \alpha_{t,t'}) \leq c \\[-1.2em]
&\forall t : \sum_k x_{t,k} = 1 \\[-0.3em]
&\forall k \in [1,m], \forall t:  0 \leq x_{t,k} \leq 1,
\end{align*}
where $c$ represents the cost; $W^{(t)}$ is the type's $t$ total load ($W^{(t)} = \sum_i p_i^{(t)}$); and $x_{t,k}$ are decision variables specifying the fraction of type $t$'s load to be allocated to machine $k$. (4) For mixed instances, the LP defined above solved for each cluster of compatible types on $m$ machines: note that this lower bound is loose, as effectively assumes $m \cdot K$ machines.

\subsection{Results}

% ====================================================================
% EXPERIMENTS: conference version
% ====================================================================
\fi

\begin{figure*}[!t]
\vspace{-2em} %Note to the editor: these negative vspaces in the figure fix the aligment of figures vs captions
\centering
\subfloat{\raisebox{-0.5cm}{\rotatebox[origin=t]{90}{{\tiny max cost (normalized by LB) }}}%%%
}
\hspace{0.1cm}
\addtocounter{subfigure}{-1}%
\subfloat[small compatible]{{\includegraphics[width=.30\textwidth]{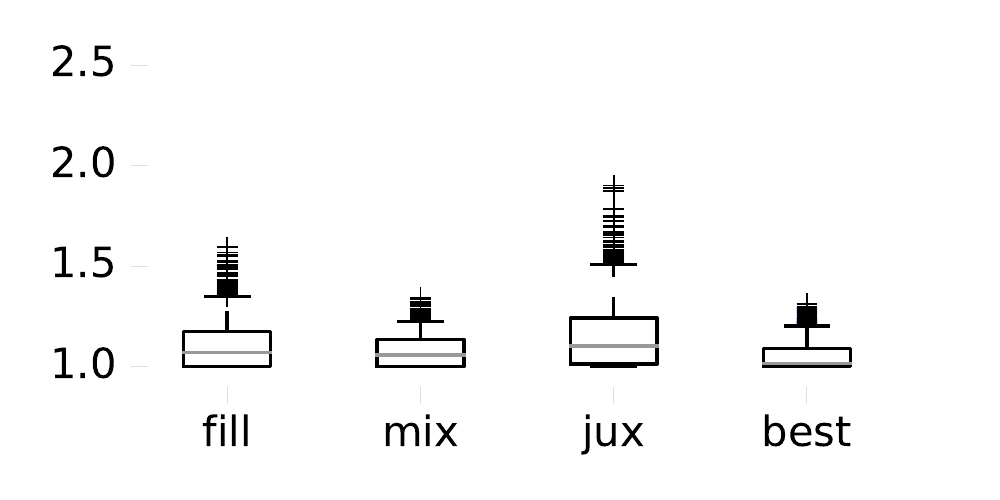} }}%
\hspace{-0.5cm}
\subfloat[small mixed]{{\includegraphics[width=.30\textwidth]{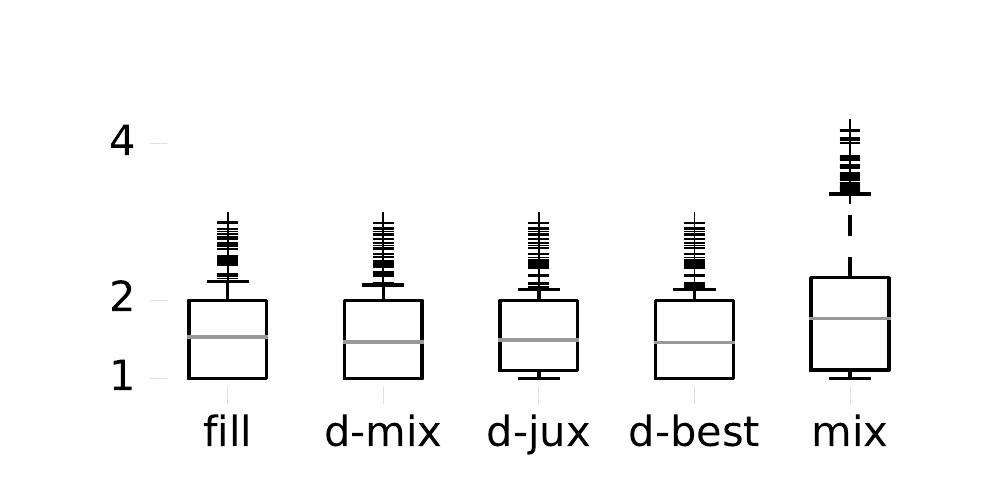} }}
\subfloat[small incompatible]{{\includegraphics[width=.30\textwidth]{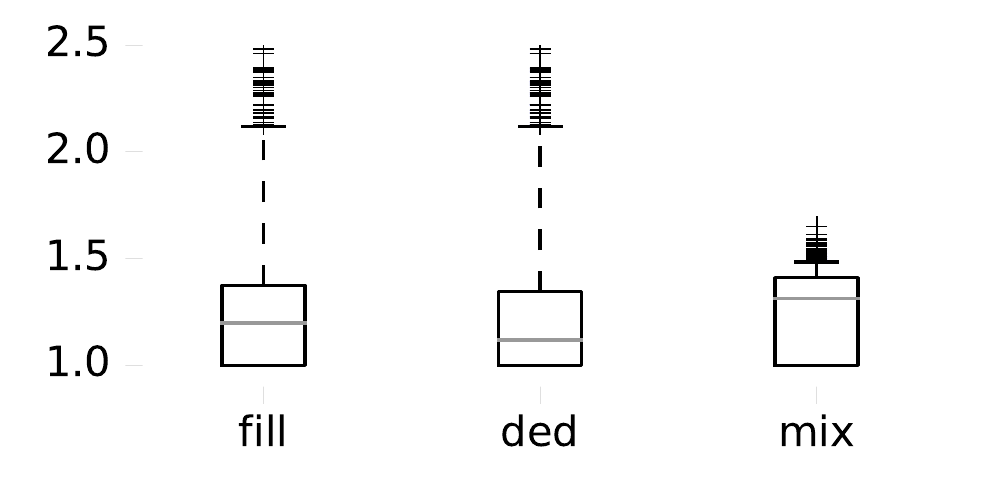} }}%

\vspace{-4em}
%\subfloat{\raisebox{1.2cm}{\rotatebox[origin=t]{90}{{\footnotesize max cost (normalized by LB)}}}   %%%
%}
%\hspace{0.1cm}
%\addtocounter{subfigure}{-1}%
\subfloat[large compatible]{{\includegraphics[width=.30\textwidth]{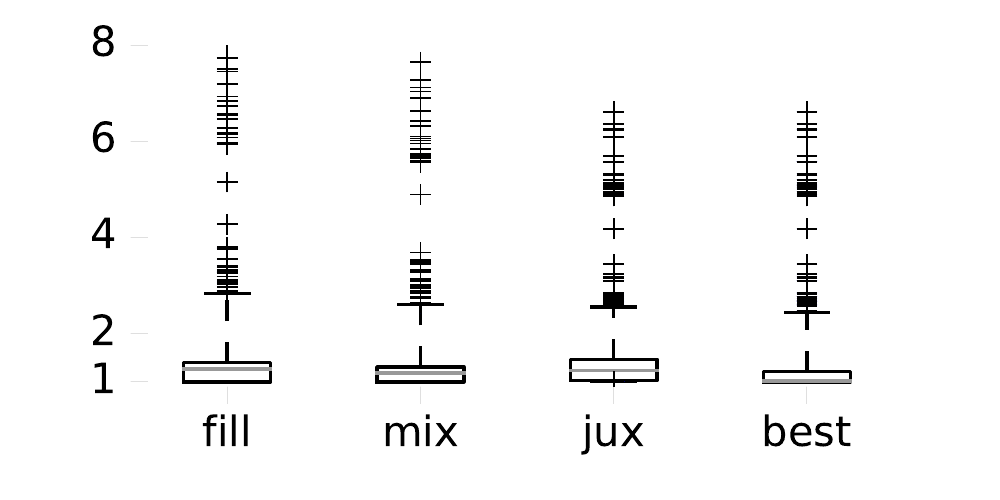} }}%
\hspace{-0.5cm}
\subfloat[large mixed]{{\includegraphics[width=.30\textwidth]{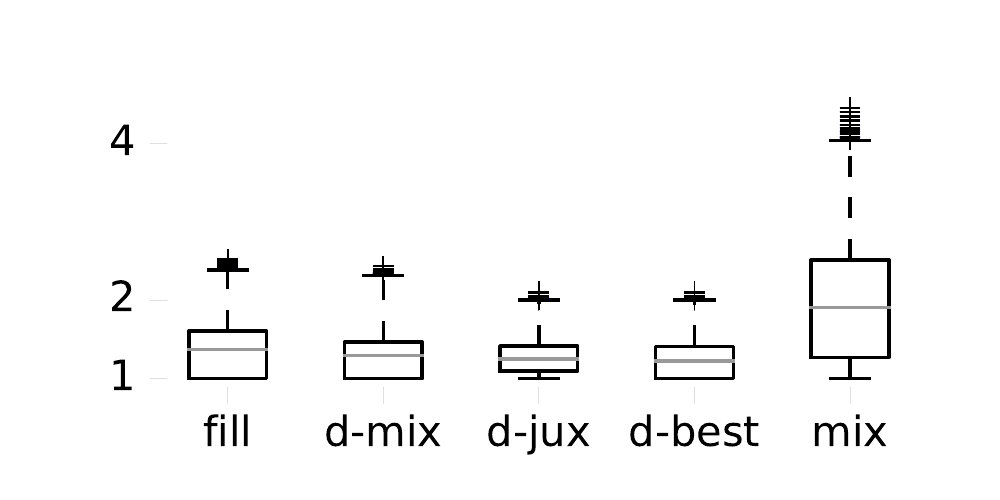} }}
\subfloat[large incompatible]{{\includegraphics[width=.30\textwidth]{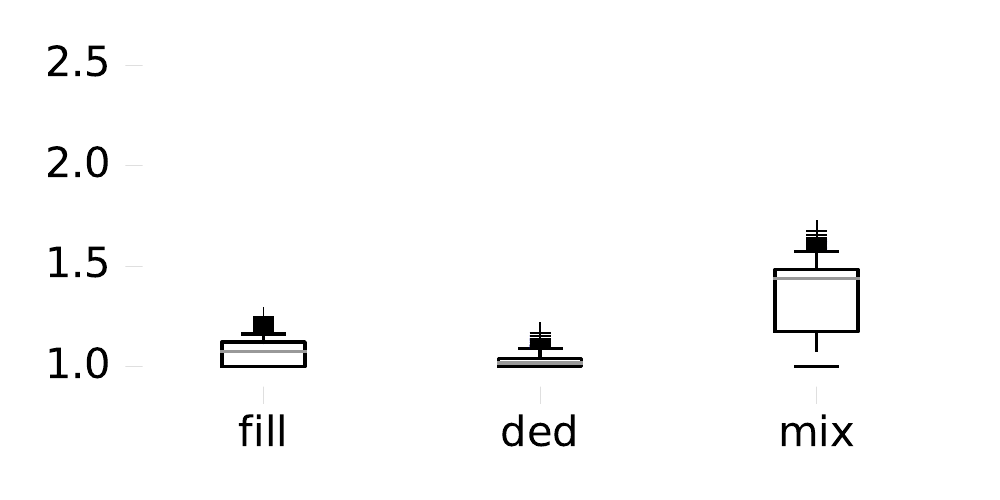} }}%
\vspace{-.5em}\caption{The maximum cost of the solutions returned by various heuristics normalized by the lower bound. All instances. In boxplots the middle line represents the median, the box spans between the 25th and the 75th percentile, the whiskers span between the 5th and the 95th percentile, and the asterisks show all the remaining points (outliers).}
\label{fig:cost-alpha-size}
\end{figure*}

\ifarxiv

\begin{figure*}[!t]
\vspace{-2em}
\centering
\subfloat{\rlap{{\tiny\hspace{0.6cm} machines $\longrightarrow$}}\raisebox{0.5cm}{\rotatebox[origin=t]{90}{{\tiny max cost}}}   %%%
}
\subfloat{\raisebox{0.5cm}{\rotatebox[origin=t]{90}{{\tiny (normalized by LB)}}} %%%
}
\hspace{0.1cm}
\addtocounter{subfigure}{-1}%
\subfloat[compatible]{{\includegraphics[width=.30\textwidth]{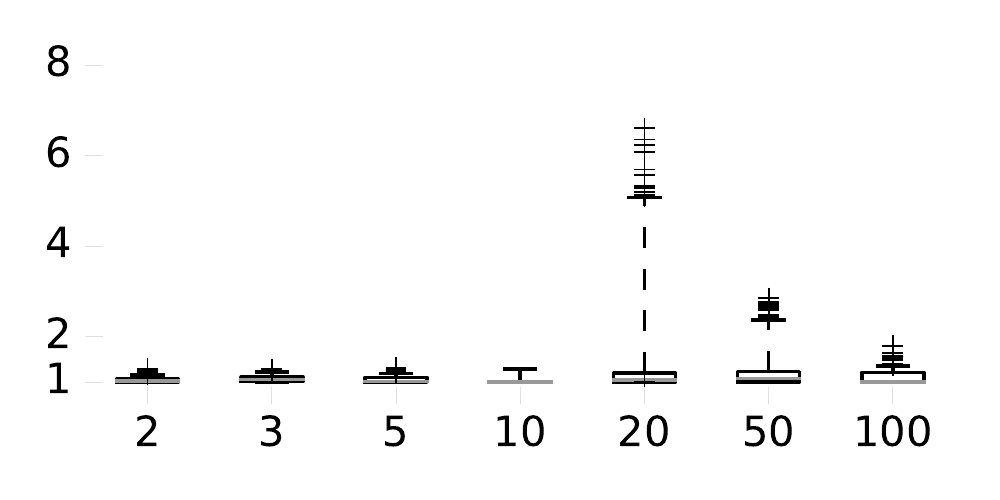} }}%
\hspace{-0.5cm}%
\subfloat[mixed]{{\includegraphics[width=.30\textwidth]{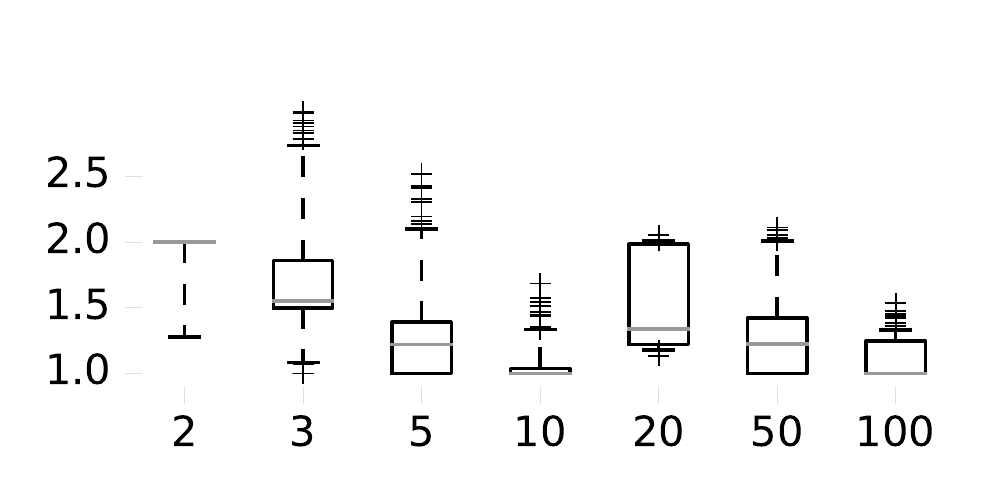} }}
\subfloat[incompatible]{{\includegraphics[width=.30\textwidth]{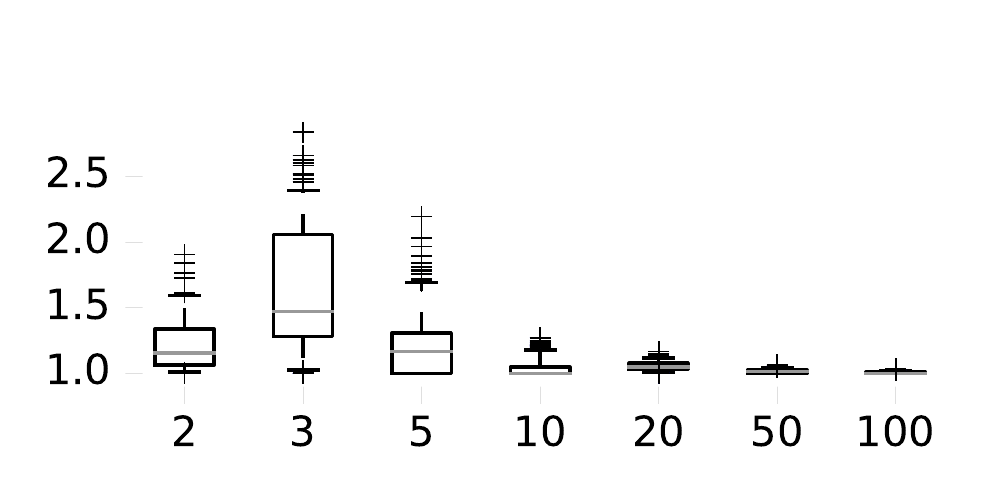} }}%
\caption{The maximum cost of the solutions returned by {\sc GreedyDedicated/SchedMixed} normalized by the lower bound. All instances.}
\label{fig:cost-by-procs}
\end{figure*}

\begin{figure}[!t]
\vspace{-1.5em}\centering
\subfloat{\raisebox{0.5cm}{\rotatebox[origin=t]{90}{{\tiny max cost}}}   %%%
}%%%
\subfloat{\raisebox{0.5cm}{\rotatebox[origin=t]{90}{{\tiny (normalized by LB)}}} %%%
}
\hspace{-0.1cm}
\addtocounter{subfigure}{-2}%
\subfloat[large compatible, LP solver completed]{{\includegraphics[width=.20\textwidth]{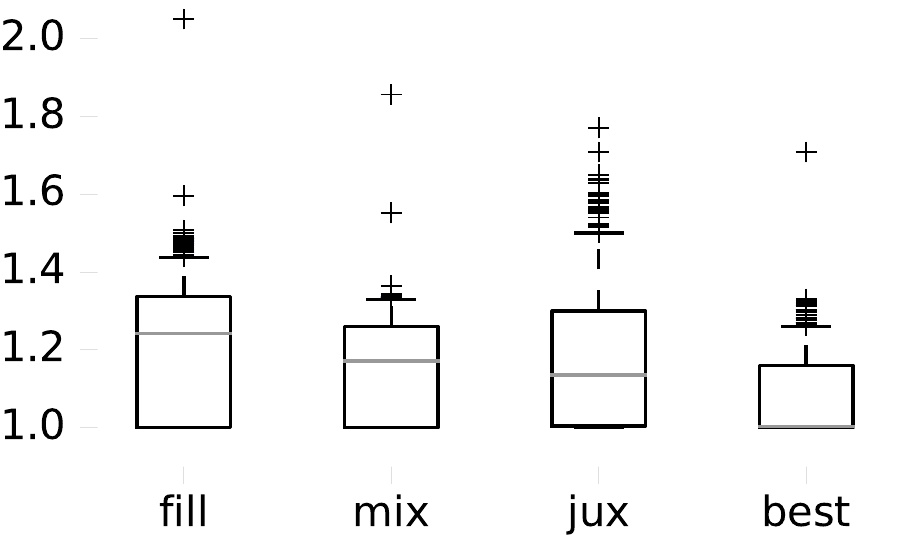} }}%
\hspace{.2cm}
\subfloat[{\sc SchedBest}, mixed instances, by number of proc.]{{\includegraphics[width=.20\textwidth]{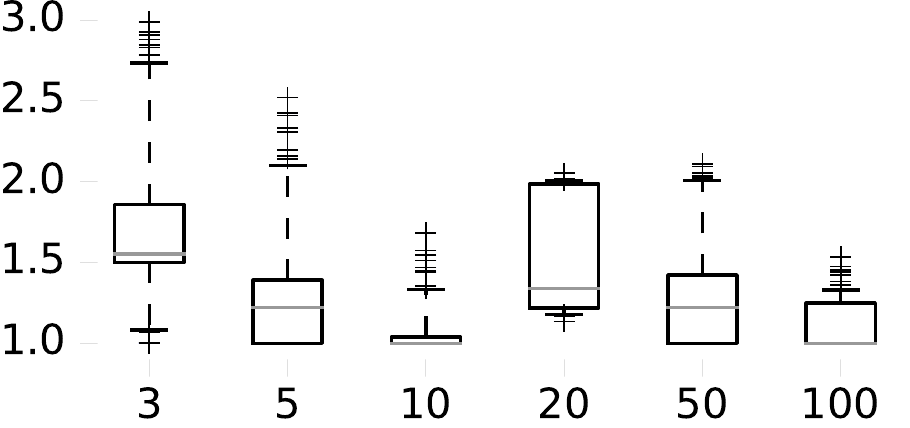} }}
\vspace{-.5em}\caption{Two cases illustrating problems with the lower bound.}\vspace{-2em}
\label{fig:cost-lp-finished}
\end{figure}
\fi

Figure~\ref{fig:cost-alpha-size} presents the normalized cost (scores on clashing and incompatible instances are exactly the same).
\ifarxiv
\subsubsection{Problems with the LP lower bound} 
We had two kinds of problems with the lower bound, both resulting in underestimation of the optimal solution and thus overestimation of the cost of our algorithms.
First, the LP solver we used (python-scipy) often failed on large compatible instances: on 6\% of $T=3$ and 70\% of $T=4$ instances.
On these instances, we used $p_{\max}$, which resulted in a lower bound that might underestimated the actual solution cost, and therefore an possible overestimation of the cost of the solutions returned by our algorithms.
When these instances are filtered out (Figure~\ref{fig:cost-lp-finished}.a), the number of outliers drops. To reduce the effect of the outliers on statistics, we discuss medians, rather than means, in the sequel.

The second problem is that the lower bound underestimates the cost of \emph{mixed} solutions. This can be seen in Figure~\ref{fig:cost-lp-finished}.b: the median score of {\sc SchedBest} is worst for two cases with highest average load per processor ($m=3$ on small instances and $m=20$ on large instances); other algorithms behave similarly. Figure~\ref{fig:cost-by-procs} shows that the largest 95th percentile was for 20 machines.
\subsubsection{Relative performance of  the algorithms} 
\else
We had two kinds of problems with the lower bound (details in~\cite{pascual2017maxcost-techreport}), both resulting in underestimation of the optimal solution and thus overestimation of the cost of our algorithms.
First, the LP solver we used (python-scipy) often failed on large compatible instances: on 6\% of $T=3$ and 70\% of $T=4$ instances. Second, the LP lower bound underestimates the optimal cost of mixed instances. In such problematic instances, $p_{\max}$ was often used, which resulted in a lower bound that might significantly underestimate the cost of the optimum. To reduce the effect of such outliers, we discuss medians, rather than means, in the sequel.
\fi

Overall, all algorithms have similar performance and the performance is close to the lower bound except in mixed instances. On average,
{\sc GreedyDedicated}, {\sc SchedMixed} and {\sc SchedJux} produce schedules with lower costs than {\sc Fill Greedy} (note that {\sc SchedMixed} and {\sc SchedJux} are used directly for compatible instances, and as sub-procedures for {\sc GreedyDedicated} for the mixed instances); and {\sc BestSchedule} optimizes even further. All the results below are statistically-significant (two sided paired t-test, p-values smaller than 0.0001).

Incompatible coefficients isolate the difference between {\sc Fill Greedy}, {\sc GreedyDedicated}, and {\sc SchedMixed} (as all clusters have a single type, neither {\sc Fill Greedy} nor {\sc GreedyDedicated} allocate different types onto a single machine).
Results clearly show that sharing machines ({\sc SchedMixed}) leads to higher costs.
{\sc GreedyDedicated} produces allocations with the lowest cost: its median costs are 1.02 for large instances and from 1.12 for small instances.

Compatible coefficients isolate the difference between {\sc Fill Greedy}, {\sc SchedMixed} and {\sc SchedJux}. On the average, {\sc SchedMixed} produces schedules with a lower cost than {\sc SchedJux} (medians are 1.06 for small instances and 1.18 for large). However, {\sc BestSchedule}, choosing for each instance the best out of {\sc SchedMixed} and {\sc SchedJux} has even lower costs (1.01 for both small and large), demonstrating the need to occasionally use {\sc SchedJux}.
\ifarxiv
For small instances, {\sc Fill Greedy} has the second-lowest cost (median is 1.07, vs. 1.10 for {\sc SchedMixed}). In contrast, for large instances, {\sc SchedMixed} has the second-lowest cost (median is 1.24 vs 1.27 for {\sc SchedMixed}).
\fi

Finally, mixed coefficients test both aspects; however, the scores of all algorithms are higher due to an imprecise lower bound. {\sc GreedyDedicated} using {\sc BestSchedule} dominates other algorithms with medians 1.46 for small instances and 1.22 for large ones. While the numerical values are higher, we still clearly see the advantage of using type-aware algorithms, as {\sc SchedMixed} (used without {\sc GreedyDedicated}) has a significantly higher median score (1.77 for small instances, 1.91 for large).
\ifarxiv
Here, {\sc Fill Greedy} is dominated by other algorithms. In small instances, {\sc SchedMixed} produces allocations of slightly lower cost than {\sc SchedJux} ($1.06$ vs $1.10$). In large instances, {\sc SchedJux} produces allocations of slightly lower cost ($1.016$ vs $1.006$).
However, {\sc SchedJux} has more outliers in this case (Figure~\ref{fig:cost-by-procs}).
\fi

Due to space constraints, we do not present results in function of the number of tasks or the number of types. However, we have not found any strong dependencies between these variables and the results of our algorithm (apart from slightly --- up to 1.18 --- higher medians for 500 and 1000-task instances for compatible coefficients, caused by the LP problems discussed above).

Our results clearly show that using $P||C_{\max}$ algorithms without regarding types ({\sc SchedMixed}) is dominated by approaches considering types: using dedicated machines for $\alpha>1$ or, in some $\alpha<1$ instances, merging schedules of different types.

\section{Related Work}
%\vspace{-.5em}
% \fixme{We did not describe what we have done in HiPC. I think that it is not a pb since we wrote ``(we studied an alternative, utilitarian model in our previous work~\cite{pascual2015partition})." in the introduction, and we say that we already study this pb in HiPC?}

We introduced the side-effects performance model~\cite{pascual2015partition}, where we studied a utilitarian (min-sum) objective. We proved that the problem is NP-hard, and we showed a dominance property (for each type, there is an order of the machines such that the tasks are assigned by decreasing sizes to the machines). This allows us to give an exact polynomial time algorithm when there is a single type. For the general case, we proposed two  algorithms, which are exponential in one data of the problem (number of types, and either the number of machines or the number of admissible sizes of the tasks).

%We introduced the side-effects performance model~\cite{pascual2015partition}, where we studied a utilitarian (min-sum) objective. We proved NP-hardness for an arbitrary number of types; we showed that an optimal solution separates tasks to groups based on their types and then assigns tasks to machines by tasks' decreasing loads; and we proposed a dynamic programming algorithm exponential in the number of types and admissible sizes of tasks.

{\em Alternative models of data~center resource management.} 
A recent survey is~\cite{pietri2016mapping}. Many colocation performance models are too complex for
combinatorial results~\cite{kundu2010application,kundu2012modeling,7134716}. Schedulers rely on heuristic approaches with no formal performance guarantees~\cite{chiang2011tracon,verma2015large,bu2013interference,Jersak:2016:PSC:2851613.2851625}. In \emph{bin-packing} approaches  (e.g.,\cite{song2014adaptive,DBLP:conf/ipps/TangLRC16}), tasks are modeled as items to be packed into bins (machines) of known capacity~\cite{coffman1996approximation}. 
To model heterogeneity, bin packing is extended to vector packing: item's size is a vector with dimensions corresponding to requirements on individual resources (CPU, memory, disk or network bandwidth)~\cite{stillwell2012virtual}. 
%These are hard optimization problems: bin packing is strongly NP-hard (but has an asymptotic PTAS~\cite{vega91binpacking}), while two-dimensional vector packing does not admit an asymptotic PTAS~\cite{woeginger1997there}
Alternatively, if tasks have unit-size requirements, simpler representations can be used, such as maximum weighted matching~\cite{beaumont2013heterogeneous}.
Bin packing approaches assume that machines' capacities are crisp and that, as long as machines are not overloaded, any allocation is equally good for tasks.
In our model, machines' capacities are not crisp---instead, tasks' performance gradually decreases with increased load.

{\em Statistical approaches.}  %More systems approaches: 
Bobroff et al.~\cite{bobroff_dynamic_2007} uses statistics of the past CPU load of tasks (CDF, autocorrelation, periodograms) to predict the load in the ``next'' time period; then they use bin packing to calculate a partition minimizing the number of used bins subject to a constraint on the probability of overloading servers. Di et al.~\cite{di2015optimization} analyze resource sharing for streams of tasks to be processed by virtual machines. Sequential and parallel task streams are considered in two scenarios. When there are sufficient resources to run all tasks, optimality conditions are formulated. When the resources are insufficient, fair scheduling policies are proposed. 

{\em Analysis of effects of colocation.} Studies showing performance degeneration when colocating data~center tasks include \cite{kambadur2012measuring, koh2007analysis,xu2013bobtail}. \cite{podzimek15cpupinning} analyze the performance of colocated CPU-intensive benchmarks; and \cite{kim15corunneraff} measures performance of colocated HPC applications. Our $\alpha_{t,t'}$ coefficients are similar to their interference/affinity metrics. Additionally \cite{kim15corunneraff} shows a greedy allocation heuristics, but they don't study its worst-case performance nor the complexity of the problem.
\vspace{-1em}

\section{Conclusion}
We considered a problem of optimally allocating tasks to machines in the side-effects performance model. Performance of a task depends on the load of other tasks colocated on the same machine. We use a linear performance function: the influence of tasks of type $t'$ is their total load times a coefficient $\alpha_{t',t}$, describing how compatible is $t'$ with performance of $t$. We minimize the maximal cost. We prove that this NP-hard problem is hard to approximate if there are many types. However, handling a limited number of types is feasible: we show a PTAS and a fast approximation algorithm, as well as a series of heuristics
\ifarxiv
that are approximation algorithms for two types.
\else
(we prove their approximation ratios for two types in the accompanying technical report~\cite{pascual2017maxcost-techreport}).
\fi
We simulate allocations resulting from algorithms on instances derived from one of Google clusters.
Our simulations show that algorithms taking into account types lead to significantly lower costs than non-type algorithms.

Our results show a possible way to adapt to data centers the large body of work in scheduling, which development has been often inspired by advances in HPC platforms. We deliberately chose to study a fundamental problem, a minimal extension to $P||C_{\max}$. We envision that results for more realistic variants of data center resource management problem, taking into account release dates, non-clairvoyance or on-line, can be taken into account similarly as they are considered in classic scheduling.

We are also working on validating our model by a systems study. We are developing an extension for kubernetes that collects and correlates performance metrics reported by containers to derive the size/coefficient performance model.

%\noindent{\bf Acknowledgements} Krzysztof Rzadca thanks Pawel Janus for his help in processing
\noindent{\bf Acknowledgements} We thank Paweł Janus for his help in processing the Google cluster data. This research has been partly supported by
%a Google Faculty Research Award,
a Polish National Science Center grant Sonata (UMO-2012/07/D/ST6/ 02440), and a Polonium grant (joint programme of the French Ministry of Foreign Affairs, the Ministry of Science and Higher Education and the Polish Ministry of Science and Higher Education).

\bibliographystyle{splncs03}
\bibliography{biblipse}
\end{document}